\documentclass[12pt,english]{article}
\usepackage[T1]{fontenc}
\usepackage[latin9]{inputenc}
\usepackage[letterpaper]{geometry}
\usepackage{color}
\usepackage{float}
\usepackage{textcomp}
\usepackage{url}
\usepackage{amsthm}
\usepackage{amsmath}
\usepackage{amssymb}

\makeatletter

\providecommand{\tabularnewline}{\\}

 \theoremstyle{definition}
 \newtheorem*{defn*}{Definition}
\theoremstyle{plain}
\newtheorem{thm}{Theorem}
  \theoremstyle{plain}
  \newtheorem{prop}[thm]{Proposition}

\usepackage{amsmath}
\usepackage{amssymb}
\usepackage{url}
\usepackage{listings}
\usepackage{url}

\makeatother

\usepackage{babel}

\begin{document}

\title{Simplifying products of fractional powers of powers}

\author{David R. Stoutemyer%
\thanks{dstout at hawaii dot edu%
}}
\maketitle
\begin{abstract}
Most computer algebra systems incorrectly simplify\[
\dfrac{z-z}{\dfrac{\sqrt{w^{2}}}{w^{3}}-\dfrac{1}{w\sqrt{w^{2}}}}\]
to 0 rather than to 0/0. The reasons for this are:

1. The default simplification doesn't succeed in simplifying the denominator
to 0.

2. There is a rule that 0 is the result of 0 divided by anything that
doesn't simplify to either 0 or 0/0.

Many of these systems have more powerful optional transformation and
general purpose simplification functions. However that is unlikely
to help this example even if one of those functions can simplify the
denominator to 0, because the input to those functions is the result
of \textsl{default} simplification, which has already incorrectly
simplified the overall ratio to 0. Try it on your computer algebra
systems!

This article describes how to simplify products of the form $w^{\alpha}\left(w^{\beta_{1}}\right)^{\gamma_{1}}\cdots\left(w^{\beta_{n}}\right)^{\gamma_{n}}$
correctly and well, where $w$ is any real or complex expression and
the exponents are rational numbers.

It might seem that correct good simplification of such a restrictive
expression class must already be published and/or built into at least
one widely used computer-algebra system, but apparently this issue
has been overlooked. Default and relevant optional simplification
was tested with 86 examples on 5 systems with $n=1$. Using a spectrum
from the most serious flaw being a result that is not equivalent to
the input somewhere to the least serious being not rationalizing a
denominator when that doesn't cause a more serious flaw, the overall
percentage of most flaw types is alarming:

\ 

\noindent \medskip{}
\begin{tabular}{|c||c|c|c|c|c|c|c|c|}
\hline 
\negthinspace{}flaw:\negthinspace{} & $\not\equiv$ & \negthinspace{}0-recognition\negthinspace{} & \negthinspace{}${\mathrm{cancelable}\atop \mathrm{singularity}}$\negthinspace{} & \negthinspace{}${\mathrm{extra}\atop \mathrm{factor}}$\negthinspace{} & \negthinspace{}excessive $\left|\gamma_{k}\right|$\negthinspace{} & \negthinspace{}$\neg$\,canonical\negthinspace{} & \negthinspace{}$\neg$\,idempotent\negthinspace{} & $\frac{\cdots}{\sqrt{\cdots}}$\tabularnewline
\hline
\hline 
\%: & 11 & 50 & 25 & 16 & 32 & 39 & 0.4 & 6\tabularnewline
\hline
\end{tabular}
\end{abstract}

\section{Introduction\label{sec:Introduction}}

\begin{flushright}
\textsl{{}``When you are right you cannot be too radical};''\\
-- Martin Luther King Jr.
\par\end{flushright}

First, a few crucial definitions:
\begin{defn*}
\textbf{\,Default simplification} is what a computer-algebra system
does to a standard mathematical expression when the user presses \fbox{E\textsc{nter}}
or \fbox{S\textsc{hift}} \fbox{E\textsc{nter}}, using factory-default
mode settings without enclosing the expression in an optional transformational
function such as $\mbox{expand}(\ldots)$, $\mbox{factor}(\ldots)$,
or $\mbox{simplify}(\ldots)$.
\end{defn*}
Default simplification is the minimal set of transformations that
a system does routinely. Default simplification is called \textsl{evaluation}
in \textsl{Mathematica}\textsuperscript{®} and in some other systems.
Any fixed set of default transformations is likely to omit ones that
are wanted in some situations and to include ones that are unwanted
in other situations. Therefore:
\begin{itemize}
\item Most systems also provide optional transformations done by a function
such as $\mathrm{expand}\left(\ldots\right)$ or by assigning a certain
value to a control variable such as $\mathrm{trigExpand}\leftarrow\mathrm{true}$.
\item Some systems provide a way to disable default transformations. For
example the Maxima assignment $\mathtt{simp:false}$ suppresses most
simplification, whereas the Maxima $\mathrm{box}(\ldots)$, \textsl{Mathematica}
$\mathrm{Hold}[\ldots]$ and Maple $\mathrm{freeze}(\ldots)$ functions
suppress most or all transformations on their argument.\end{itemize}
\begin{defn*}
Simplification is \textbf{idempotent} for a class of input expressions
if simplification of the result (by the same default or optional transformations)
yields the same result.
\end{defn*}

\begin{defn*}
A \textbf{conveniently cancelable singularity} is a removable singularity
that can be removed exactly by functional identities such as $\sin(2w)\equiv2\sin(w)\cos(w)$
together with transformations such as a common denominator followed
by factoring out the gcd of any resulting numerator and denominator,
then using the law of exponents $w^{\mu}w^{\nu}\rightarrow w^{\mu+\nu}$.
\end{defn*}
For example, $z^{3}z^{-2}\rightarrow z$, $\sin(2z)/\sin(z)\rightarrow2\cos(z)$,
and \[
\frac{1}{c\,\left(cx-1\right)}+\frac{1}{c}\rightarrow\frac{x}{cx-1},\]
which cancels the removable singularity at $c=0$, leaving the non-removable
singularity along the hyperbola $cx=1$. However the removable singularity
in $\sin(z)/z$ is not \textsl{conveniently} cancelable because it
can't be canceled exactly except inconveniently by means such as introducing
the piecewise function\[
\dfrac{\sin(z)}{z}\rightarrow\begin{cases}
1, & \mathrm{if}\; z=0,\\
\dfrac{\sin(z)}{z}, & \mathrm{otherwise,}\end{cases}\]
or the infinite series\[
\dfrac{\sin(z)}{z}\rightarrow\sum_{k=0}^{\infty}\dfrac{(-1)^{k}z^{2k}}{(2k+1)!}.\]

\begin{defn*}
A \textbf{nested power produc}t is an expression or a sub-expression
of the form\begin{equation}
w^{\alpha}\left(w^{\beta_{1}}\right)^{\gamma_{1}}\cdots\left(w^{\beta_{n}}\right)^{\gamma_{n}},\label{eq:DefinitionOfNestedPowerProduct}\end{equation}
 with $n\geq1$, rational exponents, and $\alpha$ possibly 0 or 1.
\end{defn*}
This article describes simple algorithms that can be used in default
and/or optional transformations to simplify nested power products
correctly and well. The abstract presents one example of why this
is important.

Default and relevant optional transformations for \textsl{Derive}\textsuperscript{®}
6.00, TI-CAS version 3.10%
\footnote{The computer algebra embedded in a succession of TI handheld calculators,
Windows and Macintosh computers has no name independent of the product
names, the most recent of which is TI-Nspire\textsuperscript{tm}.%
}, Maxima 5.24.0, Maple\textsuperscript{tm} 15.00 and \textsl{Mathematica}
8.0.4.0 were tested on 86 examples for the simplest case where $n=1$.
The table in the Abstract shows the overall percentages for each of
eight different decreasingly serious flaw types described in Section
\ref{sec:ListOfGoals}.

Those large percentages for the six most serious kinds of flaws are
alarming, and so are many corresponding percentages for each of the
five systems.%
\footnote{I am guilty as a coauthor of \textsl{Derive} and TI-computer algebra.%
} Wikipedia currently lists 29 other computer algebra systems, and
I strongly suspect that most or all of them also have substantial
room for improvement in this regard.

Here is an outline of the rest of the article: Section \ref{sec:More-important-definitions}
defines three more crucial terms. Section \ref{sec:ListOfGoals} describes
eight prioritized goals for results that are nested power products,
why they are important, and the reasons for their priorities. Section
\ref{sec:Experimental-results} describes the tables of results at
the end of this article and how the listed result flaws were measured.
Section \ref{sec:Four-alternative-forms} describes four good forms
for nested power products and how to obtain them:
\begin{enumerate}
\item Form 1 merely standardizes the outer fractional exponents to the interval
$(-1,1)$ in a way that doesn't introduce removable singularities,
but instead tends to reduce their magnitude -- perhaps completely.
\item Form 2 further reduces many outer fractional exponents to $[-1/2,1/2]$
in a way that cancels as much of any removable singularity as can
be done without resorting to form 4. Form 2 is an improvement on form
1 at the expense of more computation.
\item Form 3 absorbs $w^{\alpha}$ into one of the nested powers just prior
to display if $w^{\alpha}$ can thus be totally absorbed, giving a
result with one less factor. Form 3 is an aesthetic improvement on
form 2 at the expense of more computation.
\item Form 4 completely cancels any cancelable singularity and nicely collapses
all of the exponents into a single unnested exponent. However, this
form often entails a complicated unit magnitude piecewise constant
factor that is -1 raised to a complicated exponent. Unsophisticated
users might be baffled by this factor, and even sophisticated users
might abhor the mess. However, this form must be addressed because
it can occur in input, it is valuable for some purposes, and some
computer algebra systems generate this form for some inputs. 
\end{enumerate}
Section \ref{sec:Unimplemented-extensions} suggests how to extend
the algorithms to recognize syntactically different but equivalent
instances of $w$ in nested power products and how to extend the algorithms
to some kinds of non-numeric exponents. Section \ref{sec:Summary}
is an overall summary. The Appendix lists about one page of \textsl{Mathematica}
rewrite rules that implement most of the third result form. Tables
of results and their flaw numbers for the five systems and for the
rewrite rules are at the end of the article.

\section{\label{sec:More-important-definitions}More key definitions}

In this article:
\begin{itemize}
\item Unless stated otherwise, an \textbf{indeterminate} is a variable that
has no assigned value, rather than a result such as 0/0. 
\item Any finite or infinite-magnitude complex value can be substituted
for indeterminates in expressions.
\item Fractional powers and square roots denote the principal branch.%
\footnote{By default some systems assume that indeterminates represent \textsl{real}
values and/or use the \textsl{real} branch wherein for reduced integers
$m$ and $n$, $(-1)^{m/n}\rightarrow1$ for $m$ even, and $(-1)^{m/n}\rightarrow-1$
for $m$ and $n$ odd. However, most computer algebra systems provide
a way to force the principal branch if it isn't the default -- and
to declare that an indeterminate is complex if that isn't the default.%
}\end{itemize}
\begin{defn*}
A \textbf{canonical form} for a class of expressions is one for which
all equivalent expressions in the class are represented uniquely.
\end{defn*}
Canonical forms help make cancellations of equivalent sub-expressions
automatic. For example, if a computer algebra system always makes
arguments of functional forms such as $\sin(\ldots)$ canonical, then
an input sub-expression such as $\sin\left((x+1)^{2}\right)-\sin\left(x^{2}+2x+1\right)$
automatically simplifies to 0 rather than remaining unchanged as a
bulky land mine that might make a subsequent result incorrect. Without
canonical arguments, recognition of cryptically similar factors and
terms requires costly tests such as determining if the difference
in corresponding arguments can be simplified to 0. This might happen
every time the same two functional forms meet during processes such
as expansion of an integer power of a sum containing two sines, which
can be often. In contrast, canonical arguments permit a much faster
mere syntactic comparison of functional forms.

As discussed in \cite{Brown,MosesSimplification,Stoutemyer10Commandments},
canonical forms are unnecessarily costly and rigid for the entire
class of expressions addressed by general-purpose computer algebra
systems. However, canonical forms are acceptable and good for default
simplification of some simple classes of \textsl{irrational} \textsl{sub-expressions}
such as nested power products.
\begin{defn*}
\textbf{Zero-recognizing simplification} for a class of expressions
is simplification for which all expressions in the class equivalent
to 0 are transformed to 0.
\end{defn*}
As illustrated by the example in the Abstract, a failure to recognize
that a sub-expression is equivalent to 0 can lead to dramatically
incorrect results. Therefore it is desirable for default simplification
to have at least a zero-recognition property. It has been proven impossible
to guarantee this even for some rather simple classes of irrational
expressions, but a strong effort should be made to achieve at least
zero recognition for as broad a class of expressions as is practical.
\begin{defn*}
\textbf{Candid simplification} produces results that are not equivalent
to an expression that visibly manifests a simpler expression class.
\end{defn*}
For example, in a candid result there are no superfluous variables,
degree magnitudes are not larger than necessary, there are no unnecessary
irrational sub-expressions, and irrationalities are nested no more
deeply than necessary. Thus without being as rigidly constrained as
canonical forms, candid simplification yields more desirable properties
than mere zero-recognizing simplification.
\begin{defn*}
In this article \textbf{undefined} means an unknown point in the entire
infinite complex plane, such as the result of 0/0.%
\footnote{It is of course audacious to define undefined. Although unnecessary
for this article, systems could usefully also\vspace{-8pt}

\begin{itemize}
\item display 0/0 as 0/0 rather than a vague controversial word such as
{}``undefined'', and\vspace{-5pt}

\item contract functions of 0/0 to strict subsets of the complex plane wherever
possible, such as $\arg(0/0)\rightarrow(-\pi,\pi]$. Having $\arg(0/0)\rightarrow0/0$
snatches defeat from the jaws of compromise. Try this on your systems!
Many systems throw an error, which is worse because it requires even
amateur authors of functions to know about all the potential throws,
catch them or vet to prevent them, and respond appropriately to make
their functions robust.
\end{itemize}
}
\end{defn*}
\vspace{1pt}

\begin{defn*}
A \textbf{conveniently representable }subset of the infinite complex
plane is one that is reasonably representable using constant expressions
extended by sets, intervals and the symbol $\infty$.
\end{defn*}
Conveniently representable proper subsets of the infinite complex
plane are regarded here as defined. Particular computer algebra systems
might not be able to represent the full range of possibilities, but
this article is suggesting what \textsl{should} be done as well as
reporting the current situation. These ideas are discussed in more
detail in \cite{StoutemyerUsefulNumbers}, but for this article the
major defined subset of interest that isn't a single point is the
result of $u/0$ for any particular non-zero complex constant $u$.
This result should be some representation of complex infinity. Among
many other benefits it permits the correct computation\begin{eqnarray*}
\dfrac{1}{1+\dfrac{1}{0}} & \rightarrow & 0.\end{eqnarray*}
Does your computer algebra system do this?
\begin{itemize}
\item For \textsl{Mathematica}, $1/0\rightarrow\mathtt{ComplexInfinity}$.
\item For \textsl{Derive}, with its default real domain, $1/0\rightarrow\pm\infty$.
\item For TI-CAS, regrettably $1/0\rightarrow\mathrm{undef}$.
\item Maxima and Maple inconveniently throw an error.
\end{itemize}
When a proper subset of the infinite complex plane isn't conveniently
representable, then the next best thing is to degrade it to 0/0. However,
that shouldn't be done for subsets that are as easily represented
as complex infinity.

If finite or infinite magnitude complex numbers are substituted for
all of the indeterminates in an \textsl{unsimplified input} expression,
then that \textsl{input} expression is undefined at that point if
and only if the result is 0/0.
\begin{defn*}
A \textbf{generalized limit} is the set of uni-directional limits
of an input expression from all possible directions in the complex
plane.
\end{defn*}
When the generalized limit of an input expression at a conveniently
cancelable singularity is a conveniently representable proper subset
of the entire infinite complex plane, then this article regards it
as not only acceptable but \textsl{commendable} to cancel the singularity
and thereby produce a \textsl{result} expression whose substitutional
value is that conveniently representable subset at that point.

Reasons for this attitude about mathematics software include:
\begin{itemize}
\item Otherwise the results tend to be unacceptably complicated.%
\footnote{Canceling a gcd occasionally increases bulk significantly, such as
$(x^{99}-1)/(x-1)\rightarrow x^{98}+x^{97}+\cdots+x+1,$ but the algorithms
described here consider only \textsl{syntactic} cancellation, which
always decreases bulk.%
}
\item There is a high likelihood that the physical problem is actually continuous
there too -- Nature abhors a removable singularity. Removable singularities
are often an artifact of the modeling such as using a polar or spherical
coordinate system.
\item Cancelable singularities are often a result of an \textsl{unnecessary}
previous transformation unavoidably done by a system (such as inappropriate
rationalization of a denominator) or a result of a previous transformation
such as monic normalization, a tangent half angle substitution, or
expansion into partial fractions deemed necessary to obtain an anti-derivative.
\item Cancellation to simplify nested power products is consistent with
quiet transformations such as $w/w\rightarrow1$ that are currently
\textsl{unavoidable} in most computer algebra systems;
\item Symbolic cancelation tends to reduce rounding errors near removable
singularities for subsequent substitution of floating-point numbers.
\end{itemize}
However, this transformation of expressions has the composability
consequence that substitution of numeric values doesn't necessarily
commute with simplification. To accommodate either treatment of expressions,
computer algebra systems could and should build in \textsl{provisos}
such as {}``$|\: w\neq0$'' that are optionally attached automatically
to intermediate and final results containing canceled removable singularities,
as suggested in \cite{CorlessAndJeffreyProvisos,Stoutemyer10Commandments}.
Meanwhile, implementers who do not want to completely cancel cancelable
singularities for simplifying nested power products can adapt the
algorithms presented here to merely reduce the magnitude of cancelable
singularities, such as $\left(z^{2}\right)^{5/2}/z^{3}\rightarrow\left(z^{2}\right)^{3/2}/z$
rather than transforming all the way to $z\sqrt{z^{2}}$.

\section{\label{sec:ListOfGoals}A list of goals for simplifying nested power
products}

The most important concern is correctness, followed by candidness,
then aesthetics and compliance with custom. More specifically, here
is a list of desirable but partially conflicting goals for simplifying
nested power product and their differences, in decreasing order of
importance:
\begin{enumerate}
\item The result should be equivalent to the input \textsl{wherever the
input is defined}. (It is acceptable for the result to be a generalized
limit of the input where the input is 0/0.)
\item A linear combination of two or more equivalent nested power products
should simplify to a multiple of a single power product -- or to 0
if the linear combination is equivalent to 0.
\item Let the \textbf{net exponent} of $\left(w^{\beta_{k}}\right)^{\gamma_{k}}$
be\[
\triangle_{k}:=\beta_{k}\gamma_{k},\]
and for a product of nested powers of $w$ let the \textbf{total positive
nested exponent} and the \textbf{total negative nested exponen}t be\begin{eqnarray}
\triangle_{+} & := & \sum_{k=1}^{n}\max\left(\triangle_{k},0\right),\\
\triangle_{-} & := & \sum_{k=1}^{n}\min\left(\triangle_{k},0\right).\end{eqnarray}
When possible, use the transformation\begin{equation}
\left(w^{\beta_{k}}\right)^{\gamma_{k}}\rightarrow w^{m_{k}\beta_{k}}\left(w^{\beta_{k}}\right)^{\gamma_{k}-m_{k}}\label{eq:ShiftingIntegersInOneNestedPower}\end{equation}
with appropriate integers $m_{k}$ to minimize $\min\left(\max\left(\alpha,0\right)+\triangle_{+},\:-\min\left(\alpha,0\right)-\triangle_{-}\right)$,
thus canceling as much of any removable singularity as is possible
by this means.
\item When possible, \textsl{fully} absorb the $w^{\alpha}$ into the nested
powers of $\left(w^{\beta_{k}}\right)^{\gamma_{k}}$ to have fewer
factors.
\item Otherwise use transformation (\ref{eq:ShiftingIntegersInOneNestedPower})
to minimize $\triangle_{+}-\triangle_{-}$ to minimize the contributions
of the troublesome nested powers.
\item Inputs that are equivalent where both are defined should produce the
same (canonical) result.
\item Results should be idempotent: Reapplying the same default or optional
simplification to the result should leave it unchanged.
\item To help achieve goal 6, rationalize a denominator in a nested power
product when this doesn't introduce a removable singularity or increase
its magnitude.
\end{enumerate}
A larger numbered goal should not be fulfilled if the only way to
fulfill it is to violate a smaller numbered goal. For example, fulfillment
of goals 5 or 8 can often violate goals 1, 2 and/or 3. 

The reasons for this ranking of the goals are:
\begin{enumerate}
\item A violation of goal 1 is most unsatisfactory because it is a result
that is not equivalent to the input everywhere the input is defined.
For example, if expression $w$ can be 0, then rationalizing the denominator
of $1/\sqrt{w}$ to give $\sqrt{w}/w$ makes an input that is a well-defined
complex infinity at $w=0$ become 0/0 there. A more serious example
is the mal-transformation $(w^{-2})^{-1/2}\rightarrow(w^{2})^{1/2}$,
because the two sides differ along the entire positive and negative
imaginary axis. For example $(i^{-2})^{-1/2}=-i$, whereas $(i^{2})^{1/2}=i$. 
\item The example in the Abstract shows the importance of zero recognition.
For example if default simplification of one nested power product
produces $\sqrt{z}/z$ and default simplification of another nested
power product produces the equivalent expression $z/\sqrt{z}$, then
the latter violates goal 8 and together they violate goal 6. These
violations are minor; but if default simplification doesn't simplify
their difference to 0, then that is a violation of goal 2, which is
serious.
\item A violation of goal 3 is next most serious because it is a squandered
opportunity to improve the result by canceling a conveniently cancelable
singularity and thereby making the result have the limiting value
at $w=0$ rather than be undefined there. For example,\begin{eqnarray}
\dfrac{\left(w^{2}\right)^{5/2}}{w}\;\rightarrow & \dfrac{w^{2}\left(w^{2}\right)^{3/2}}{w} & \rightarrow\; w\left(w^{2}\right)^{3/2},\label{eq:FirstEgOfGoal3}\\
w\left(\dfrac{1}{w^{2}}\right)^{5/2}\rightarrow & \dfrac{w}{w^{2}}\left(\dfrac{1}{w^{2}}\right)^{3/2} & \rightarrow\;\dfrac{1}{w}\left(\dfrac{1}{w^{2}}\right)^{3/2}.\label{eq:SecondEgOfGoal3}\end{eqnarray}

\item A violation of goal 4 is more complicated than need be. For example,
most people would agree that $w^{4}\left(w^{2}\right)^{1/2}$ is more
complicated than $\left(w^{2}\right)^{5/2}$, which has one less factor.
\item Goal 5 is important because when there is more than one factor of
the form $\left(w^{\beta}\right)^{\gamma_{k}}$, there might be more
than one way to distribute only some of $w^{\alpha}$ into the nested
powers. In contrast if $\gamma_{k}$ is a half-integer power then
there are only two ways to minimize $\left|\gamma_{k}\right|$ by
factoring an integer power of $w^{\alpha}$ out of $(w^{\alpha})^{\gamma}$,
or only one way for other fractional powers. Moreover, unnested exponents
are less specialized and can therefore interact more freely with other
factors in a product. For example, for intermediate results (\ref{eq:FirstEgOfGoal3})
and (\ref{eq:SecondEgOfGoal3}),\begin{eqnarray*}
w\left(w^{2}\right)^{3/2} & \rightarrow & w^{2}\left(w^{2}\right)^{1/2},\\
\dfrac{1}{w}\left(\dfrac{1}{w^{2}}\right)^{3/2} & \rightarrow & \dfrac{1}{w^{2}}\left(\dfrac{1}{w^{2}}\right)^{1/2}.\end{eqnarray*}

\item For a given $w$, the above goals tend to yield the most concise possible
nested power product in terms of $w$. Therefore it is better to have
the consistency of having two inputs that are equivalent where they
are defined return the same most concise form. More importantly a
canonical form for nested power product sub-expressions greatly facilitates
achieving goal 2 -- a major benefit for very little effort.
\item Without idempotency, an unaware user could obtain inconsistent results,
and a cautious aware user would have to re-enter such results as inputs
until they cycle or stop changing.
\item Goal 8 complies with the custom of rationalizing denominators and
helps achieve canonicality goal 6. For example,\begin{eqnarray*}
\dfrac{w}{\sqrt{w^{2}}}\;\rightarrow & \dfrac{w\sqrt{w^{2}}}{w^{2}} & \rightarrow\;\dfrac{\sqrt{w^{2}}}{w}.\end{eqnarray*}
But rationalization should not be done at the expense of lower-numbered
goals. For example,\begin{eqnarray*}
\dfrac{1}{w\left(w^{2}\right)^{2/3}} & \not\rightarrow & \dfrac{\left(w^{2}\right)^{1/3}}{w^{3}},\end{eqnarray*}
because although it reduces the absolute value of the outer exponent
(goal 5), it violates goal 1 by making an input that is complex infinity
at $w=0$ become a result that is 0/0 there.
\end{enumerate}

\section{\label{sec:Experimental-results}Important information about the
tables}

Tables at the end of this article show the results that occurred for
each example with each system and with the Appendix rewrite rules.
In all of the tables the goal numbers in the Section \ref{sec:ListOfGoals}
goals that aren't satisfied but could be satisfied without violating
a lower-numbered goal are listed beside each result. Unmet goal numbers
1 and 2 are boldface to emphasize their extreme seriousness.

\subsection{Examples, test protocol and table interpretation}

Tables \ref{Flo:MathematicaDefaultTable} through \ref{Flo:MaximaFullratsimpTableAndMapleSimplifyTable}
report default and relevant optional transformation results for \textbf{test
family 1}: multiplying $w^{m}$ by $\left(w^{2}\right)^{n+2/3}$ for
successive integer $m=-3$ through 3 in combination with successive
integer $n=-3$ through 2. For comparison with results that meet all
of the goals, Table \ref{Flo:AlgorithmAtable} has corresponding results
for form 3 described in Section \ref{sec:Four-alternative-forms},
as produced by the one page of \textsl{Mathematica} rewrite rules
listed in the Appendix.

Tables \ref{Flo:MathematicaDefaultSqrtReciprocalTable} through \ref{Flo:MapleSimplifySqrtReciprocal}
report default and relevant optional transformation results for \textbf{test
family 2}: multiplying $w^{m}$ by $\left(w^{\boldsymbol{-}2}\right)^{n+1/2}$
for the same combinations of $m$ and $n$. For comparison with results
that meet all of the goals, Table \ref{Flo:AlgorithmASqrtReciprocalTable}
has corresponding form 3 results produced by the Appendix rewrite
rules.

To help assess compliance with goal 8, Table \ref{Flo:WonSqrtWSquaredTable}
compares results for all of the systems with the Appendix rewrite
rules on the particularly simple input $w/\sqrt{w^{2}}$ . Some of
the examples in test families 1 and 2 also test this goal.

Table \ref{Flo:Form1MinusForm4Table} tests only whether or not the
expression\[
\dfrac{\sqrt{w^{2}}}{w}-\left(-1\right)^{^{\frac{1}{2}\left(\arg(w^{2})-2\arg(w)\right)/\pi}}\]
simplifies to 0. This is the difference between equivalent expressions
in form 3 and form 4. This is a more difficult but not impossible
problem. All systems fail -- including the Appendix rewrite rules,
which do not address this issue.

Here is how compliance with the goals was assessed:
\begin{enumerate}
\item Most of the results that violated goal 1 did so only at $w=0$. However,
the goal 1 violations in Tables \ref{Flo:MaximaDefaultSqrtReciprocal}
and \ref{Flo:MaximaFullratsimpSqrtRecip} instead or also are not
equivalent to the input where it is defined along the entire positive
and negative imaginary axis. This is caused by an outlaw of exponents:
transforming $\left(w^{-\lambda}\right)^{\mu}$ to $\left(w^{\lambda}\right)^{-\mu}$
for fractional $\mu$, which is not valid along these semi-axes.
\item For test family 1, each input is equivalent to the input two rows
down and one column left wherever both are defined, and their omnidirectional
limits are identical wherever one of the inputs is 0/0. Therefore
to assess compliance with goal 2 (zero-recognition), for every entry
in the table I computed the difference $w^{m}\left(w^{2}\right)^{n+2/3}-w^{m+2}\left(w^{2}\right)^{n-1+2/3}$
or the optional transformation thereof and the difference $w^{m}\left(w^{2}\right)^{n+2/3}-w^{m-2}\left(w^{2}\right)^{n+1+2/3}$,
then considered it a flaw for the entry if either of these two differences
was non-zero. Thus compliance with this goal is \textsl{not} discernible
from merely inspecting the result entries. Compliance is a property
of the default simplification or optional transformation when given
the difference of two non-identical but equivalent nested power products.
For test family 2, each input is equivalent to the input two rows
down and one column \textsl{right }wherever both are defined, so I
did an analogous test for that. It is of course possible for an entry
to pass these limited tests but fail for more widely separated equivalent
inputs.%
\footnote{This happens for Tables \ref{Flo:MaximaDefaultSqrtReciprocal} and
\ref{Flo:MaximaFullratsimpSqrtRecip}: \textsl{All} of the columns
would exhibit flaw 2 if one of the two equivalent expressions was
always taken from the correct results in columns 3 or 4. The results
are not equivalent to the inputs for columns 1, 2, 5 and 6, so the
only reason the difference simplified to 0 for columns 1 and 6 was
the subtraction of incorrect but identical results -- an instance
where two wrongs make a right.%
} For Table \ref{Flo:WonSqrtWSquaredTable} the result is equivalent
to $\sqrt{w}/w$ and $w^{2}/(w^{2})^{3/2}$, so I tested whether or
not the corresponding two differences or optional transformation thereof
simplified to 0. Table \ref{Flo:Form1MinusForm4Table} tests zero
recognition directly, and only than -- but with only \textsl{one}
particular difference in equivalent forms rather than only two.
\item To comply with goal 3 without violating lower-numbered goals, a result
$w^{\hat{\alpha}}\left(w^{\beta}\right)^{\hat{\gamma}}$ should have\begin{eqnarray*}
 & \alpha=0\:\vee\:\\
 & \left(\mathrm{sign}\left(\hat{\alpha}\right)=\mathrm{sign}\left(\beta\hat{\gamma}\right)\:\wedge\:\left|\hat{\gamma}\right|<1\right)\:\vee\:\\
 & (\mathrm{\alpha\neq0\:\wedge\: sign}\left(\hat{\alpha}\right)\neq\mathrm{sign}\left(\beta\hat{\gamma}\right)\:\wedge\:\left|\hat{\gamma}\right|<1\:\wedge\:\\
 & \min(\left|\beta\hat{\gamma}\right|,\left|\hat{\alpha}\right|)\:\leq\:\min(\left|\beta\left(\hat{\gamma}-\mathrm{sign}\left(\hat{\gamma}\right)\right)\right|,\,\left|\hat{\alpha}+\beta\,\mathrm{sign}\left(\hat{\gamma}\right)\right|)).\end{eqnarray*}

\item To comply with goal 4, $\hat{\alpha}$ should not be an integer multiple
of $\beta$.
\item To comply with goal 5 without violating lower numbered goals,\[
\hat{\alpha}=0\;\vee\;\left(\mathrm{sign}\!\left(\hat{\alpha}\right)\!=\!\mathrm{sign}\!\left(\beta\hat{\gamma}\right)\:\wedge\:\left|\hat{\gamma}\right|\!<\!1\right)\;\vee\;\left(\mathrm{\alpha\!\neq\!0\:\wedge\: sign}\!\left(\hat{\alpha}\right)\!\neq\!\mathrm{sign}\!\left(\beta\hat{\gamma}\right)\:\wedge\:\left|\hat{\gamma}\right|\!\leq\!\frac{1}{2}\right).\]

\item For compliance with goal 6, every result entry for test family 1 should
be \textsl{identical} to the entry 2 rows down and one column left,
whereas every result entry for test family 2 should be identical to
the entry 2 rows down and one column right. To equally assess the
top two rows, the bottom two rows, the leftmost column and the rightmost
column, I computed extra neighbors bordering those shown. When there
were differences, I did not penalize the best displayed results for
equivalent entries unless their was a better displayed result one
column and one or two rows outside the table. However, I did penalize
\textsl{all} of the not-best members for equivalent entries. I similarly
tested the result in Table \ref{Flo:WonSqrtWSquaredTable} against
the equivalent expressions $\sqrt{w}/w$ and $w^{2}/(w^{2})^{3/2}$.
It is of course possible for an entry to pass these tests but fail
for more widely separated equivalent inputs.
\item To test compliance with goal 7, I resimplified each result with either
default simplification or the optional transformation used for the
original input, then checked for identical results. Compliance with
this goal is \textsl{not} discernible from merely inspecting the result
entries.
\item To test compliance with goal 8, I manually rationalized the results
having a fractional power in the denominator and $\alpha\neq0$ by
multiplying the numerator and denominator by $(w^{-2})^{1/3}$ for
test family 1 or $\sqrt{w^{2}}$ for test family 2. It was counted
as a flaw if and only if that forced rationalization did not introduce
a removable singularity or increase its magnitude.
\end{enumerate}

\subsection{Remarks about particular results.}

Maxima also has a relevant $\mathrm{rat}(\ldots)$ function. For these
examples, it generally produces the same result as default simplification,
except that fractional powers are represented as an integer power
of a reciprocal power -- or an integer power of $\sqrt{\ldots}$ for
half-integer powers. Thus a default result $w^{3}\left(w^{2}\right)^{5/3}$
would instead be $w^{3}((w^{2})^{1/3})^{5}$, and a default result
$w^{3}\left(w^{2}\right)^{5/2}$ would instead be $w^{3}\sqrt{w^{2}}^{5}$.
The standard definition of $u^{m/n}$ for reduced integers $m$ and
$n$ \textsl{is} $\left(u^{1/n}\right)^{m}$, which is consistent
with the alternate definition $e^{\ln\left(u\right)m/n}$.%
\footnote{This is \textsl{not} generally equivalent to $\left(u^{m}\right)^{1/n}$:
{}``be faithful to your roots'' -- Mason Cooley.%
} Consequently, the Maxima $\mathrm{rat}\left(\ldots\right)$ function
makes the standard interpretation of the result more explicit at the
expense of clutter. Nonetheless, it might be helpful as a precursor
to semantically substituting a new expression for $\left(w^{2}\right)^{1/3}$
via syntactic substitution in a expression containing $\left(w^{2}\right)^{m/3}$
for several different integer $m$. For the sake of brevity, results
are not included for the $\mathrm{rat}(\ldots)$ function because
its flaws were very nearly identical to default simplification, regarding
$((w^{2})^{1/n})^{m}$ as $(w^{2})^{m/n}$.

\textsl{Mathematica}, Maxima and Maple also respectively have relevant
PowerExpand{[}...{]}, radcan(...) and simplify(..., symbolic) functions.
However, they always transform $(w^{\beta})^{\gamma}$ to $w^{\beta\gamma}$,
which is not equivalent along entire rays from $w=0$. I didn't test
these functions because their purpose is presumably partly to allow
these risky unconditional transactions for consenting adults. However,
these three systems, \textsl{Derive} and TI computer algebra also
have \textsl{safe} ways to enable such desired transformations, when
justified, by declaring, for example, that certain variables are real
or positive.

\begin{flushright}
{}``\textsl{\ldots{} a man who thought he could somehow pull up
the root without affecting the power.}''\\
--adapted from Gilbert K. Chesterton
\par\end{flushright}

To make sure that $w$ is regarded as a complex variable and the principal
branch is used rather than the real branch:
\begin{itemize}
\item All of the \textsl{Derive} results follow a prior declaration $w:\in\mathtt{Complex}$.
\item All of the TI-CAS results necessarily used $w\_$ rather than $w$
to manifestly declare it as a complex indeterminate. However for consistency
$w\_$ is displayed in all of the tables as $w$ because Table \ref{Flo:TIandMapleDefaultSqrtReciprocal}
is shared with Maple for brevity.
\item All of the Maxima results followed a prior assignment $\mathtt{domain:complex}$
and a prior declaration $\mathtt{declare(w,complex)}$.
\end{itemize}
As illustrated by Tables \ref{Flo:MapleSimplifySqrtReciprocal} through
\ref{Flo:Form1MinusForm4Table}, the Maple $\mathrm{simplify}(\ldots)$
function expresses half-integer powers of squares or of reciprocals
of squares using the Maple $\mathrm{csgn}(\ldots)$ function defined
by\begin{equation}
\mathrm{csgn}(w):=\begin{cases}
1, & \mathrm{if}\;\Re(w)>0\:\vee\;\Re(w)=0\;\wedge\;\Im(w)\geq0,\\
-1, & \mathrm{otherwise}.\end{cases}\label{eq:DefinitionOfCsgn}\end{equation}
The right side of this definition is a simplified special instance
of form 4, for which $\mathrm{csgn}(w)$ is a convenient abbreviation
for those familiar with it.%
\footnote{Jeffrey \cite{Jeffrey} uses the unwinding function to generalize
csgn to a $C_{n}$ that works \textsl{for all} fractional powers.
If and when implemented in Maple, that will avoid unwelcome mixtures
of $\mathrm{csgn}(\ldots)$ with other form 4 notations for results
containing both half-integer and other nested powers.%
}

Regarding Table \ref{Flo:Form1MinusForm4Table}:
\begin{itemize}
\item \textsl{Mathematica} (hence also the Appendix rewrite rules) did the
automatic transformation\[
(-1)^{\left(1/2\right)\left(\mathrm{Arg}\left[w^{2}\right])-2\mathrm{Arg}\left[w\right]\right)/\pi}\rightarrow i^{\left(\mathrm{Arg}\left[w^{2}\right])-2\mathrm{Arg}\left[w\right]\right)/\pi}.\]
Although it eliminates the $1/2$ factor from the exponent in this
case, it does so at the expense of candidness by introducing $i$
into an expression that is real for all $w$.
\item For TI-CAS the $\arg$ function is spelled {}``angle'' and regrettably
angle(0) is returned unchanged rather than transforming to 0. Therefore
the input was\[
\dfrac{\sqrt{w\_^{2}}}{w\_}-\left(-1\right)^{\mathrm{when}\left(w\_=0,\:0,\:(1/2)(\mathrm{angle}(w\_^{2})-2\mathrm{angle}(w\_))/\pi\right).}\]
As indicated in Table \ref{Flo:Form1MinusForm4Table}, the real power
of -1 in the input was changed to an imaginary power of $e$ in the
result. This has the candidness disadvantage of introducing $i$ into
an expression that is real for all real $w\_$.
\item For \textsl{Derive} the arg function is spelled {}``phase'' and
regrettably phase(0) returns $\pi/2\pm\pi/2$, which denotes an unknown
element of $\left\{ 0,\pi\right\} $, which are the only two possibilities
for \textsl{real} arguments. Therefore the input was\[
\dfrac{\sqrt{w^{2}}}{w}-\left(-1\right)^{\mathrm{IF}\left(w=0,\;0,\;(1/2)(\mathrm{PHASE}(w{}^{2})-2\mathrm{PHASE}(w))/\pi\right).}\]

\item For Maple, the $\arg$ function is spelled {}``argument'', and for
Maxima it is spelled {}``carg''.
\end{itemize}
If you are interested in results for some other systems, then try
a few of the examples that are heavily flawed for most of the five
tested systems.%
\footnote{If you are familiar enough with those systems, then most of them probably
have a quick way to generate all of the results for test families
1 and 2 by entries analogous to the following one for \textsl{Mathematica}:\[
\mathtt{Table\,}[\mathtt{Table\,}[w^{j}(w^{2})^{k},\left\{ k,\,-7/3,\,8/3\right\} ,\left\{ j,\,-3,\,3\right\} \;//\mathtt{TableForm}\]
I am interested in knowing your results.%
} First do whatever is necessary so that fractional powers use the
principal branch and $w$ is regarded as a complex indeterminate.

\section{\label{sec:Four-alternative-forms}Four alternative forms}

Section \ref{sec:Introduction} explains the reasons for four separate
forms.

\subsection{Form 1: Reduction of outer fractional exponents to (-1, 1)\label{subWidth2Interval}}
\begin{defn*}
For $x\in\mathbb{R}$ the \textbf{integer part function}\[
\mathrm{Ip}\left(x\right):=\begin{cases}
\left\lfloor x\right\rfloor , & \mathrm{if}\; x\geq0,\\
\left\lceil x\right\rceil , & \mathrm{otherwise}.\end{cases}\]

\end{defn*}
\ 

\vspace{-12pt}

\begin{defn*}
For $x\in\mathbb{R}$ the \textbf{fractional part function} $\mathrm{Fp}(x):=x-\mathrm{Ip}(x)$.\end{defn*}
\begin{prop}
For $\beta\in\mathbb{Q}$, $\gamma\in\mathbb{Q}-\mathbb{Z},$ and
arbitrary expression $w\in\mathbb{C}$, \begin{equation}
\left(w^{\beta}\right)^{\gamma}\equiv w^{\mathrm{\beta\, Ip}\left(\gamma\right)}\left(w^{\beta}\right)^{\mathrm{Fp}\left(\gamma\right)}.\label{eq:TransformationToMinus1One}\end{equation}
\end{prop}
\begin{proof}
We have\begin{equation}
\left(w^{\beta}\right)^{\gamma}\equiv\left(w^{\beta}\right)^{\mathrm{Ip}\left(\gamma\right)}\left(w^{\beta}\right)^{\mathrm{Fp}\left(\gamma\right)}\label{eq:IpPlusFP}\end{equation}
 because: 

1. With $\gamma\in\mathbb{Q}-\mathbb{Z}$, $\mathrm{Ip}(\gamma)=0\;\vee\;\mathrm{sign}\left(\mathrm{Ip}(\gamma)\right)=\mathrm{sign}\left(\mathrm{Fp}(\gamma)\right)$.

2. For any expression $u\in\mathbb{C}$ and $r_{1},r_{2}\in\mathbb{Q}\;|\; r_{1}=0\;\vee\;\textrm{sign}\, r_{1}=\textrm{sign}\, r_{2}$,\begin{equation}
u^{r_{1}+r_{2}}\equiv u^{r_{1}}u^{r_{2}},\label{eq:uR1PlusR2Equal}\end{equation}
even at $u=0$ with $r_{1}$ and $r_{2}$ both negative, making both
sides of (\ref{eq:uR1PlusR2Equal}) be complex infinity.

3. By Proposition \ref{pro:IntegerExponent} we also have $\left(w^{\beta}\right)^{\mathrm{Ip}\left(\gamma\right)}\equiv w^{\mathrm{\beta\, Ip}\left(\gamma\right)}$
because $\mathrm{Ip}\left(\gamma\right)\in\mathbb{Z}$.
\end{proof}
Therefore Form 1 is simply to transform $w^{\alpha}\left(w^{\beta_{1}}\right)^{\gamma_{1}}\cdots\left(w^{\beta_{n}}\right)^{\gamma_{n}}$
toward canonicality by transforming every positive fraction $\gamma_{k}$
to the interval $(0,1)$ and every negative fraction $\gamma_{k}$
to the interval $(-1,0)$. The various $w^{\mathrm{Ip}\left(\gamma_{k}\right)\,\beta_{k}}$
are combined with the original $w^{\alpha}$, giving a transformed
expression\begin{equation}
\widehat{W}:=w^{\hat{\alpha}}\left(w^{\beta_{1}}\right)^{\hat{\gamma}_{1}}\cdots\left(w^{\beta_{n}}\right)^{\hat{\gamma}_{n}}\label{eq:DefinitionOfWHat}\end{equation}
\label{sub:DefinitionOfWHat}where $\hat{\alpha}$ might be 0.

This form 1 satisfies goals 1 and 7, while possibly contributing progress
toward goals 2, 3, 5 and 6. This form also has the advantage that
if $\left(w^{\beta}\right)^{\hat{\gamma}}$ is subsequently raised
to any power $\lambda$, then we can simplify it to the simplified
value of $\left(w^{\beta}\right)^{\gamma\lambda}$ by Proposition
\ref{pro:ExponentInMinus1To1} because $-1<\hat{\gamma}<1$. For example,\begin{eqnarray*}
\left(\left(w^{2}\right)^{3/4}\right)^{7/6} & \rightarrow & \left(w^{2}\right)^{7/8}.\end{eqnarray*}
Although there is no such thing as a free radical in computer algebra,
this transformation of each nested power is fast and easy to implement
because it occurs only for certain fractional powers of powers, which
are relatively rare, and very little work is done even when it does
occur. There is no good reason why default simplification shouldn't
do at least this much.

However, default simplification for many systems unavoidably collects
similar factors, resulting in a partial reversal of this transformation
whenever a resulting unnested exponent $\hat{\alpha}$ is identical
to one of the inner nested exponents. This happens for \textsl{Derive},
TI-CAS and \textsl{Mathematica}, but not for Maple or Maxima. With
unavoidable collection, unconditional magnitude reduction of fractional
outer exponents can lead to an infinite recursion such as\[
w\left(w^{2}\right)^{5/3}\rightarrow w^{2}\left(w^{2}\right)^{2/3}\rightarrow w\left(w^{2}\right)^{5/3}\rightarrow\cdots.\]
Therefore in the Appendix \textsl{Mathematica} rewrite rules:
\begin{itemize}
\item Transformation $\left(w^{\beta}\right)^{\gamma}\rightarrow w^{\mathrm{Ip}\left(\gamma\right)\,\beta}\left(w^{\beta}\right)^{\mathrm{Fp}\left(\gamma\right)}$
is used unconditionally only \textsl{prior} to default simplification.
\item The rewrite rules that are active \textsl{during} default and optional
transformations do not reduce the magnitude of $\gamma$ if doing
so would give an unnested exponent $\hat{\alpha}$identical to $\beta$.
\item This transformation is used \textsl{after} default simplification
only if there is an unnested factor $w^{\alpha}$ and the transformation
would not be reversed by unavoidable collection of similar powers.
\end{itemize}
Implementations for other systems might have to overcome this difficulty
in some other way, or compromise and not always produce a form with
outer fractional exponents in the interval $(-1,1)$ when $w^{\alpha}$
can't be fully absorbed into some nested power.

\subsection{Form 2: Further reducing some outer exponents to (-1/2, 1/2{]}\label{sub:UnitWidthInterval}}

Form 2 is form 1 supplemented by an additional transformation.

Expression $\widehat{W}$ given by definition (\ref{eq:DefinitionOfWHat})
is equivalent to expression $W$ everywhere that $W$ is defined,
because at the only questionable point $w=0$:
\begin{enumerate}
\item Expressions $W$ and $\widehat{W}$ are both 0 if $\alpha\geq0$ and
all of the $\beta_{k}\gamma_{k}$ are positive.
\item Otherwise expression $W$ and $\widehat{W}$ are both complex infinity
if $\alpha\leq0$ and all of the $\beta_{k}\gamma_{k}$ are negative.
\item Otherwise if $\hat{\alpha}\geq0$ and all $\hat{\beta}_{k}\hat{\gamma}_{k}>0$,
then $W$ is 0/0 but $\widehat{W}$ has improved to 0.
\item Otherwise if $\hat{\alpha}\leq0$ and all $\hat{\beta}_{k}\hat{\gamma}_{k}<0$,
then $W$ is 0/0 but $\widehat{W}$ has improved to complex infinity.
\item Otherwise both $W$ and $\widehat{W}$ are 0/0. However, the magnitude
of the multiplicity of the removable singularity is less for $\widehat{W}$
if for any $\gamma_{k}$, $\left|\gamma_{k}\right|\geq1$.
\end{enumerate}
Expression $\widehat{W}$ is canonical in cases 1 through 4, but not
necessarily for case 5. For example,
\begin{enumerate}
\item The different equivalent expressions $z^{-1}\left(z^{2}\right)^{2/3}$
and $z\left(z^{2}\right)^{-1/3}$ both have outer exponents in (-1,
1). Of these two alternatives, the latter is preferable for most purposes
because the $\left|-2/3\right|<\left|4/3\right|$, making multiplicity
of the uncanceled portion of the removable singularity have a smaller
magnitude. Thus a rationalized numerator is sometimes preferable to
a rationalized denominator.
\item The different expressions $z^{-3}\left(z^{2}\right)^{1/2}$ and $z^{-1}\left(z^{2}\right)^{-1/2}$
both have outer exponents in (-1, 1), and they are equivalent wherever
the first alternative is defined. However, the latter unrationalized
denominator is preferable because the former is 0/0 at $z=0$ where
the latter is defined and equal to the complex infinity limit of the
former.
\item The different equivalent expressions $z\left(z^{2}\right)^{-1/2}$
and $z^{-1}\left(z^{2}\right)^{1/2}$ both have outer exponents in
(-1, 1), and the multiplicities of the uncanceled portion of their
removable singularity at $z=0$ are both 1. Of these two alternatives,
the latter is slightly preferable because it has a traditionally rationalized
denominator rather than a rationalized numerator.
\end{enumerate}
Thus after producing form 1 we can sometimes add 1 to a negative $\hat{\gamma}_{k}$
or subtract 1 from a positive $\hat{\gamma}_{k}$, then adjust $\alpha$
accordingly to reduce the magnitude of the overall removable singularity
-- perhaps entirely. If not, perhaps we can at least contribute toward
goals 2, 6 and 8 by rationalizing a square root in the denominator.

Let\begin{eqnarray*}
\Delta_{k} & := & \beta_{k}\hat{\gamma}_{k},\\
\Delta & := & \alpha+\Delta_{1}+\cdots+\Delta_{n}.\end{eqnarray*}

Transforming any of the $\left(w^{\beta_{k}}\right)^{\gamma_{k}}$
to $w^{m_{k}\beta_{k}}\left(w^{\beta_{k}}\right)^{\gamma_{k}-m_{k}}$
for any integer $m_{k}$ leaves $\Delta$ unchanged.

Our primary goal is, whenever possible, to make all of the $\Delta_{k}$
have the same sign and for $\alpha$ to have either the same sign
or be 0. A secondary goal is to prefer $-1/2<\hat{\gamma}_{k}\leq1/2$.
Therefore, the algorithm to convert form 1 to form 2 is:
\begin{enumerate}
\item If $\Delta>0$, then for each $\Delta_{k}<0$, add $\mathrm{sign}\left(\beta_{k}\right)$
to $\hat{\gamma}_{k}$ and subtract $\left|\beta_{k}\right|$ from
$\alpha,$ then return the result.
\item If $\Delta<0$, then for each $\Delta_{k}>0$, subtract $\mathrm{sign}\left(\beta_{k}\right)$
from $\hat{\gamma}_{k}$ and add $\left|\beta_{k}\right|$ to $\alpha$,
then return the result.
\item For each $\hat{\gamma}_{k}>1/2$, subtract 1 from $\hat{\gamma}_{k}$
and add $\beta_{k}$ to $\alpha$.
\item For each $\hat{\gamma}\leq-1/2$, add 1 to $\hat{\gamma}_{k}$ and
subtract $\beta_{k}$ from $\alpha$.
\item Return the result.
\end{enumerate}
This canonical form 2 satisfies all of the goals except for the aesthetic
goal 4.

For brevity the Appendix rewrite rules consider only one $\Delta_{k}$
at a time. This is sufficient for all of the test cases, which have
only one nested power. For an industrial-strength implementation,
each time we multiply a fractional power of a power by a product of
one or more factors, we should inspect those factors for identical
expressions $w$ and apply the above algorithm if that subset if non-empty.
The cost is $O\left(n_{c}\right)$ where $n_{c}$ is the number of
cofactors. The opportunity occurs only when multiplying a fractional
power of a power, which is rare; and the number of factors in a product
is typically quite small. Therefore it also quite reasonable to do
this in default simplification.

\subsection{\textsl{\textcolor{blue}{\label{sub:Form-3}}}Form 3: Finally, fully
absorb $w^{\alpha}$ into a fractional power if possible}

Form 3 is form 2 followed by an additional transformation.

Form 2 can result in an expression such as $z^{4}\left(z^{2}\right)^{1/2}$,
for which many users would regard $\left(z^{2}\right)^{5/2}$ as a
simpler result because it has one less factor. We can often absorb
at least some of $z^{\alpha}$ into one of the $\left(z^{\beta_{k}}\right)^{\gamma_{k}}$
by the transformation\[
z^{\alpha}\left(z^{\beta_{k}}\right)^{\gamma_{k}}\rightarrow z^{\beta_{k}\,\mathrm{Fp}\left(\alpha/\beta_{k}\right)}\left(z^{\mathbf{\beta_{k}}}\right)^{\gamma_{k}+\mathrm{Ip}\left(\alpha/\beta_{k}\right)},\]
which doesn't change the domain of definition. However, this transformation
seems inadvisable unless $\mathrm{Fp}\left(\alpha/\beta_{k}\right)=0$,
because otherwise it increases the contribution of a troublesome nested
power without reducing the number of factors. Also, this transformation
is problematic \textsl{during} intermediate computations even if $\mathrm{Fp}\left(\alpha/\beta_{k}\right)=0$,
because when there is more than one nested power, then more than one
might be eligible, making it awkward to maintain canonicality achieved
by form 2. Moreover, absorption conflicts with transformations done
to obtain form 1 or 2, thus risking infinite recursion.

A solution to this dilemma is to fully absorb $w^{\alpha}$ only just
before display -- after all other default and optional simplification.
This does have the minor disadvantage that what the user sees doesn't
faithfully represent the internal representation. However, that bridge
has already been crossed by most systems, which for speed and implementation
simplicity internally use, for example, $(\ldots)^{1/2}$ to represent
a displayed $\sqrt{\ldots}$ and $a+-1*b$ to represent a displayed
$a-b$.

When there is more than one nested power of $w$, then there might
be more than one way to absorb $\alpha$ completely into those nested
powers. For example,\begin{eqnarray}
w^{6}\left(w^{2}\right)^{1/2}\left(w^{3}\right)^{1/2}\left(w^{4}\right)^{1/2} & \equiv & \left(w^{2}\right)^{\boldsymbol{7/2}}\left(w^{3}\right)^{1/2}\left(w^{4}\right)^{1/2}\nonumber \\
 & \equiv & \left(w^{2}\right)^{1/2}\left(w^{3}\right)^{\boldsymbol{5/2}}\left(w^{4}\right)^{1/2}\label{eq:FirstExampleOfAbsorption}\\
 & \equiv & \left(w^{2}\right)^{\boldsymbol{3/2}}\left(w^{3}\right)^{1/2}\left(w^{4}\right)^{\boldsymbol{3/2}}.\nonumber \end{eqnarray}

In general, the possible resulting expressions are given by\[
\left(w^{\beta_{1}}\right)^{\gamma_{1}+m_{1}}\left(w^{\beta_{2}}\right)^{\gamma_{2}+m_{2}}\cdots\left(w^{\beta_{n}}\right)^{\gamma_{n}+m_{n}},\]
where the tuple of integers $\left\langle m_{1},m_{2},\ldots,m_{n}\right\rangle $
is a solution to the linear Diophantine equation\[
m_{1}\beta_{1}+m_{2}\beta_{2}+\cdots+m_{n}\beta_{n}=\alpha.\]

Solutions exist if and only if $\alpha$ is an integer multiple of
$\gcd\left(\beta_{1},\beta_{2},\ldots\beta_{n}\right)$, in which
case there might be a countably infinite number of tuples. However,
to avoid introducing removable singularities or increasing the magnitude
of their multiplicity, we are only interested in solutions for which
$\mathrm{sign}\left(m_{j}\beta_{j}\right)\equiv\mathrm{sign}\left(\alpha\right)$
for $j=1,2,\ldots,n$. Papp and Vizvari \cite{PappAndVizvari} describe
an algorithm for solving such sign-constrained linear Diaphantine
equations, and the \textsl{Mathematica} $\mathtt{Reduce}\left[\ldots\right]$
function can solve such equations. For example, suppose our canonical
form 2 result is\begin{equation}
z^{14}\left(z^{6/7}\right)^{1/2}\left(z^{10/7}\right)^{1/3}.\label{eq:nEqual2Example}\end{equation}

In \textsl{Mathematica}, we can determine the family of integers $m_{1}\geq0$
to add to $1/2$ and $m_{2}\geq0$ to add to 1/3 that together absorb
$z^{14}$ as follows: \begin{eqnarray*}
\mathsf{In}[1]: & = & \mathtt{Reduce\,}\left[\dfrac{6}{7}m_{1}\!+\!\dfrac{10}{7}m_{2}==14\:\;\&\&\:\;\dfrac{6}{7}m_{1}\geq0\:\;\&\&\:\;\dfrac{10}{7}m_{2}\geq0,\mathtt{\,\left\{ m_{1},m_{2}\right\} ,\, Integers}\right]\\
 &  & \quad//\mathtt{TraditionalForm}\end{eqnarray*}
\[
\mathsf{Out}[1]//\mathrm{TraditionalForm}=\left(m_{1}=3\wedge m_{2}=8\right)\,\vee\,\left(m_{1}=8\wedge m_{2}=5\right)\,\vee\,\left(m_{1}=13\wedge m_{2}=2\right)\]
Regarding the choice between alternative absorptions, canonicality
is not as important for a final displayed result as it is during intermediate
calculations where it facilitates important cancellations. However,
with more than one solution, we could choose one in a canonical way
as follows: Order the $\beta_{j}$ in some canonical way, such as
the way they order in $\left(w^{\beta_{1}}\right)^{\gamma_{1}}\cdots\left(w^{\beta_{n}}\right)^{\gamma_{n}}$,
then to choose the solution for which $m_{1}$ is smallest, with ties
broken according to which $m_{2}$ is smallest, etc.

Solution of sign-constrained linear Diophantine equations can be costly
-- probably too costly for default simplification. Consequently, the
rewrite rules in the Appendix simply absorb $w^{\alpha}$ if and only
if it can be completely absorbed into a single power of a power, in
which case the particular one is the first one encountered by the
pattern matcher. This is canonical, but it doesn't absorb $w^{\alpha}$
for examples such as (\ref{eq:nEqual2Example}). However, this transformation
is inexpensive because it is done only once in one pass over the expression
just prior to display, and the transformation requires comparing a
power of a power with its cofactors only in products where powers
of powers occur.

All but this absorption rule are automatically applied \textsl{before}
default simplification so that, for example, the input\[
\dfrac{w-w}{\dfrac{\sqrt{z^{2}}}{z^{3}}-\dfrac{1}{z\sqrt{z^{2}}}}\]
correctly simplifies to \texttt{indeterminate}, meaning 0/0, rather
than to 0.

The rewrite rules in the Appendix are not much more than the minimal
amount necessary to generate the form 3 results in Tables \ref{Flo:AlgorithmAtable}
and \ref{Flo:AlgorithmASqrtReciprocalTable}, together with the relevant
rows in Tables \ref{Flo:WonSqrtWSquaredTable} and \ref{Flo:Form1MinusForm4Table}.

\subsection{Form 4\label{sub:UnitMagnitudeFactor}: One unnested power times
a unit-magnitude factor}

Form 4 is quite different from forms 1 through 3.

A universal principal-branch formula for transforming a nested power
to an unnested power is\begin{equation}
\left(w^{\beta}\right)^{\gamma}\rightarrow\left(-1\right)^{\tau}w^{\beta\gamma},\label{eq:PowerOfPowerTransformation}\end{equation}
where\begin{equation}
\tau:=\begin{cases}
0, & \mathrm{if}\;\arg(0)=0,\\
\dfrac{\gamma\left(\arg\left(w^{\beta}\right)-\beta\arg\left(w\right)\right)}{\pi}, & \mathrm{otherwise},\end{cases}\label{eq:CorrectionExponentForPowerOfPower}\end{equation}
with short-circuit evaluation so that the {}``otherwise'' result
expression is not evaluated when the {}``if'' test is true.

The transformation given by formulas (\ref{eq:PowerOfPowerTransformation})
and (\ref{eq:CorrectionExponentForPowerOfPower}) can be derived from
the identities\begin{eqnarray}
\left|p\right| & \equiv & (-1)^{-\arg\left(p\right)/\pi}p\qquad\mathrm{for}\; p\neq0,\\
\left|q^{\alpha}\right|^{\beta} & \equiv & \left|q\right|^{\alpha\beta}.\end{eqnarray}

Notice that the \textbf{unit-polar} factor $(-1)^{\tau}$ is unit
magnitude because $\arg(\ldots)$ is always real, as are the rational
numbers $\gamma$ and $\beta$. Moreover, $(-1)^{\tau}$ is piecewise
constant with pie-shaped pieces emanating from $w=0$ because $\arg\left(w^{\beta}\right)$
and $\beta\arg\left(w\right)$ have the same derivative with respect
to $w$ everywhere they are both continuous, and each of them has
a finite number of discontinuities.

An imaginary exponential $e^{i\pi\tau}$ is an alternative to $(-1)^{\tau}$,
but it has the candidness disadvantage of introducing $i$ into a
factor that can be real and always is for the common case where the
outer exponent $\gamma$ is a half-integer.

If $\arg(0)$ is defined as 0, as it is in \textsl{Mathematica}, Maple,
and Maxima, then we can define $\tau$  more concisely and unconditionally
as\begin{equation}
\tau:=\dfrac{\gamma\left(\arg\left(w^{\beta}\right)-\beta\arg\left(w\right)\right)}{\pi}.\label{eq:UnconditionalDefinitionOfTau}\end{equation}

\begin{prop}
\label{pro:PositiveRadicand}If $w\geq0$, then $\tau=0$.\end{prop}
\begin{proof}
When $w=0$, $\tau=0$ follows immediately from expression (\ref{eq:CorrectionExponentForPowerOfPower}),
and\\
$w>0\:\Rightarrow\:\arg\left(w^{\beta}\right)=0\wedge\arg\left(w\right)=0\:\Rightarrow\:\gamma\left(\arg\left(w^{\beta}\right)-\beta\arg\left(w\right)\right)/\pi=0\:\Rightarrow\:\tau=0.$\end{proof}
\begin{prop}
\label{pro:ExponentInMinus1To1}If $-1<\beta\leq1$, then $\tau=0$.\end{prop}
\begin{proof}
$-1<\beta\leq1\:\Rightarrow\:\arg\left(w^{\beta}\right)=\beta\arg\left(w\right)\:\Rightarrow\:\gamma\left(\arg\left(w^{\beta}\right)-\beta\arg\left(w\right)\right)/\pi=0\:\Rightarrow\:\tau=0.$\end{proof}
\begin{prop}
\label{pro:IntegerExponent}If $\gamma$ is integer, then $(-1)^{\tau}=1$.\end{prop}
\begin{proof}
$\arg\left(w^{\beta}\right)$ is $\beta\arg\left(w\right)$ plus an
even integer multiple of $2\pi$. Thus when $\gamma$ is an integer,
then $\gamma\left(\arg\left(w^{\beta}\right)-\beta\arg\left(w\right)\right)/\pi$
is an even integer, making $\tau$ be an even integer, making $(-1)^{\tau}=1$.
\end{proof}
The simplification afforded by these three propositions should have
already been exploited with bottom-up default simplification, in which
case $\left(w^{\beta}\right)^{\gamma}$ will have already been simplified
to $w^{\beta\gamma}$. If it isn't, then that is another opportunity
to improve the system for very little effort.%
\footnote{Do your computer algebra system's default and optional transformations
de-nest $\left(w^{\beta}\right)^{\gamma}$ for such $\beta$, $\gamma$,
and $w$ declared non-negative?%
} Thus, because $\beta$ and $\gamma$ are explicit non-zero rational
numbers, without loss of generality this article assumes that $w$
isn't known to be nonnegative, and that $\beta\leq-1$ or $\beta>1$,
and that $\gamma$ is non-integer.

Using transformation (\ref{eq:PowerOfPowerTransformation}) on every
$\left(w^{\beta_{k}}\right)^{\gamma_{k}}$ in $W$ defined by (\ref{eq:DefinitionOfNestedPowerProduct})
then collecting powers of $-1$ gives\begin{equation}
\overline{W}=\left(-1\right)^{\sigma}w^{\alpha+\beta_{1}\gamma_{1}+\cdots+\beta_{n}\gamma_{n}},\label{eq:Form1}\end{equation}
where $\sigma$ is a simplified sum of terms of the form (\ref{eq:CorrectionExponentForPowerOfPower})
or (\ref{eq:UnconditionalDefinitionOfTau}).

The factor $\left(-1\right)^{\sigma}$ is also unit magnitude with
pie-shaped piecewise constant pieces because it is the product of
such factors. This form has two great advantages over the other three
forms:
\begin{itemize}
\item All of the exponents have been combined into a \textsl{single unnested
exponent}.
\item Cancelable singularities are always \textsl{completely} canceled.
\end{itemize}
Unfortunately this comes at the expense of a form that is usually
bulkier than the other forms

Simplification of individual piecewise expressions and combinations
of such expressions is currently rather weak in most systems, but
Carette \cite{Carette} describes a canonical form for such expressions,
so we can hope for improvement. In our case the piecewise expressions
all have the same tests. Therefore we can add all of the 0s together
and add all of the expressions involving $\arg(\ldots)$ together
into a single piecewise function. For example,\begin{multline}
\dfrac{\left(z^{2}\right)^{3/2}\left(z^{3}\right)^{4/3}}{z^{6}}\\
\rightarrow\left(\!\left(-1\right)^{\!\begin{cases}
0, & \!\mathrm{\!\! if}\:\arg z\!=\!0,\\
\frac{\frac{3}{2}\left(\arg\!\left(z^{2}\right)-2\arg z\right)}{\pi}, & \!\mathrm{\!\! otherwise}\end{cases}}\!\right)\!\negthinspace\left(\!\left(-1\right)^{\!\begin{cases}
0, & \!\mathrm{\!\! if}\:\arg z\!=\!0,\\
\frac{\frac{4}{3}\left(\arg\left(z^{3}\right)-3\arg z\right)}{\pi}, & \!\mathrm{\!\! otherwise}\end{cases}}\!\right)\! z^{\frac{3}{2}2+\frac{4}{3}3-6}\\
\rightarrow\left(\left(-1\right)^{\begin{cases}
0, & \mathrm{\!\! if}\:\arg z=0,\\
\frac{\frac{3}{2}\arg\left(z^{2}\right)+\frac{4}{3}\arg(z^{3})-7\arg z}{\pi}, & \mathrm{\!\! otherwise}\end{cases}}\right)z.\label{eq:ExampleOfForm4}\end{multline}

If $\arg(0)\rightarrow0$, then simplification of piecewise expressions
isn't an issue here and the resulting exponent of $-1$ is simply
$\left(\frac{3}{2}\arg\left(z^{2}\right)+\frac{4}{3}\arg(z^{3})-7\arg z\right)/\pi$.
However, the result is not canonical either way, because starting
with the equivalent canonical form 2,\begin{eqnarray}
\dfrac{\sqrt{z^{2}}\left(z^{3}\right)^{1/3}}{z} & \rightarrow & \left(\left(-1\right)^{\begin{cases}
0, & \mathrm{\!\! if}\;\arg z=0,\\
\frac{\frac{1}{2}\arg\left(z^{2}\right)+\frac{1}{3}\arg(z^{3})-2\arg z}{\pi} & \mathrm{\!\! otherwise}\end{cases}}\right)z,\label{eq:FullySimplifiedForm4Example}\end{eqnarray}
which has smaller magnitude coefficients. Thus for canonicality we
could precede this transformation with a transformation to form 2.
Equivalently we can adjust the coefficients of the $\arg\left(w^{\beta_{k}}\right)$
and $\arg(w)$ analogous to how we adjusted exponents to arrive at
form 2. This is preferable because it also canonicalizes expressions
of form 4 that are entered directly or generated by the system.

With pie-shaped pieces, $(-1)^{\sigma}$ can always be expressed in
the more candid canonical form\[
\begin{cases}
c_{1}, & \mathrm{if}\;-\pi<\arg w\:::\:\theta_{1,}\\
c_{2}, & \mathrm{if}\;\theta_{1\:}::\:\arg w\:::\:\theta_{2,}\\
\ldots & \ldots\\
c_{m}, & \mathrm{otherwise},\end{cases}\]
where $c_{1}$ through $c_{m}$ are unit-magnitude complex constants,
$\theta_{1}$ through $\theta_{m-1}$ are real constants in $(-\pi,\pi)$,
and each instance of {}``::'' is either {}``$<$'' or {}``$\leq$''.%
\footnote{In a degenerate case, one or more of the pieces of pie might be a
ray -- very dietetic. %
} Moreover:
\begin{enumerate}
\item When $w$ is real, then the positive and negative real axes are each
entirely within one pie slice, enabling us to simplify$\left(-1\right)^{\sigma}$
to one unconditional constant or piecewise expression of the form\[
\begin{cases}
c_{1}, & \mathrm{if}\: w::0,\\
c_{2} & \mathrm{otherwise},\end{cases}\]
where {}``::'' is one of the comparison operators {}``>'', {}``$\geq$'',
=, {}``$\leq$'', {}``<'', or {}``$\neq$''.
\item For half-integers or quarter-integer fractional powers, $\left(-1\right)^{\tau}$
can be expressed as a piecewise expression depending on the real and
imaginary parts of $w$ rather than $\arg(w)$. For example,\begin{eqnarray}
\dfrac{\left(w^{2}\right)^{1/2}}{w} & \rightarrow & \begin{cases}
\:1 & \mathrm{if}\:\Re\left(w\right)>0\vee\Re\left(w\right)\geq0\wedge\Im\left(w\right)\geq0,\\
-1 & \mathrm{otherwise};\end{cases}\label{eq:GoesToCsgn}\\
\dfrac{\left(w^{4}\right)^{1/4}}{w} & \rightarrow & \begin{cases}
1 & -\Re\left(w\right)<\Im\left(w\right)\leq\Re\left(w\right),\\
-i & \mathrm{if}\:-\Im\left(w\right)<\Re\left(w\right)\leq\Im\left(w\right),\\
-1 & \Re\left(w\right)<\Im\left(w\right)\leq-\Re\left(w\right),\\
i & \mathrm{otherwise}.\end{cases}\end{eqnarray}

\end{enumerate}
Notice that the right side of result (\ref{eq:GoesToCsgn}) is the
definition of the Maple csgn function. 

Without an abbreviation such as $\mathrm{csgn}(\ldots)$, Most implementers
will probably want to avoid form 4 as a \textsl{default} even when
$\arg\left(0\right)\rightarrow0$, because $(-1)^{\sigma}$ is likely
to be rather complicated nonetheless:
\begin{enumerate}
\item It will probably contain complicated square roots and arctangents
if the real and imaginary parts of $w$ are given as exact numbers.
\item It will probably also contain piecewise sign tests if given real and
imaginary parts that are non-numeric, such as for $w=x+iy$ with non-numeric
real indeterminates $x$ and $y$.
\item It will probably contain radicals nested at least one deep if $\arg\left(w\right)$
is a simple enough rational multiple of $\pi$.
\item Otherwise it will contain perhaps bulky sub-expressions $\arg\left(w\right)$
and $\arg\left(w^{\beta}\right)$ -- or, worse yet, expressions involving
square roots, arctangents, piecewise sign tests, and sub-expressions
of the form $\Re\left(w\right)$ and $\Im\left(w\right)$.
\end{enumerate}
As espoused by Corless and Jeffrey \cite{CorlessAndJeffrey}, expression
$\tau$ can alternatively be defined in terms of the unwinding function
$\kappa$ as:\begin{equation}
\tau:=2\gamma\kappa\left(\beta\ln w\right).\label{eq:TauViaUnwindingNumber}\end{equation}
This is more concise than definition (\ref{eq:CorrectionExponentForPowerOfPower}),
but a function that computes unwinding numbers isn't currently available
externally in most computer algebra systems. Also, unless the system
automatically transforms $\ln0$ to $-\infty$, as is done in \textsl{Mathematica}
and \textsl{Derive}, then definition (\ref{eq:TauViaUnwindingNumber})
has the same disadvantages as using $\arg\left(\ldots\right)$.%
\footnote{For TI-CAS, $\ln(0)\rightarrow\mathrm{undef}$. An error is inconveniently
thrown by Maple for $\ln(0)$ and by Maxima for log(0).%
}

\subsection{Simplifying mixtures of form 4 with form 1, 2 or 3}

If an expression contains a mixture of forms, then we should unify
the forms to facilitate collection and cancelation. For example with
$\arg(0)\rightarrow0$, the three expressions\begin{eqnarray}
 & \dfrac{\left(z^{2}\right)^{1/2}}{z},\label{eq:form4Alternative1}\\
 & (-1)^{\left(\arg\left(z^{2}\right)/2-\arg z\right)/\pi},\label{eq:form4Alternative2}\\
 & \begin{cases}
\:1 & \mathrm{if}\:\Re\left(w\right)>0\vee\Re\left(w\right)\geq0\wedge\Im\left(w\right)\geq0,\\
-1 & \mathrm{otherwise}\end{cases}\label{eq:form4Alternative3}\end{eqnarray}
are equivalent. Therefore the result of any linear combination of
them should transform either to 0 or a multiple of one of them. The
rewrite rules in the Appendix don't address this issue.

In general it is easy to transform form (\ref{eq:form4Alternative1})
to form (\ref{eq:form4Alternative2}), which is only slightly more
difficult to transform to either form (\ref{eq:form4Alternative1})
or form (\ref{eq:form4Alternative3}).

\section{\label{sec:Unimplemented-extensions}Unimplemented extensions}

\subsection{More semantic pattern matching for $w$}

The \textsl{Mathematica} pattern matcher is mostly syntactic rather
than semantic, and the rules in the Appendix do almost no transformation
of the radicand expressions $w$ or any cofactors thereof. Thus recognition
of opportunities relies mostly on the default transformations together
with any optional transformations done by the user. Consequently,
opportunities for the rules to simplify nested power products might
not be recognized for radicands that aren't indeterminates. The rules
work for most functional forms that have syntactically identical forms
for the different instances of $w$, such as\begin{eqnarray*}
\dfrac{\left(\mathrm{Log}\left[x^{2}\left(x+y\right)\right]^{2}\right)^{5/3}}{\mathrm{Log}\left[x^{2}\left(x+y\right)\right]} & \rightarrow & \mathrm{Log}\left[x^{2}\left(x+y\right)\right]\left(\mathrm{Log}\left[x^{2}\left(x+y\right)\right]^{2}\right)^{2/3}.\end{eqnarray*}
However the rules don't apply to \textsl{all} such functional form
opportunities. For example,\[
\dfrac{\left(\mathrm{Cos}[\theta]^{2}\right)^{5/3}}{\mathrm{Cos}[\theta]}\rightarrow\left(\mathrm{Cos}[\theta]^{2}\right)^{5/3}\mathrm{Sec}[\theta]\]
because default simplification transforms $\mathrm{Cos}[\theta]^{-1}$
to $\mathrm{Sec}[\theta]$.

Even more opportunities are unrecognized when $w$ is a sum. As an
example of how to overcome this, the Appendix includes one extra rule
that square-free factors radicands that are sums so that, for example,\[
\dfrac{\left(z^{2}+2z+1\right)^{5/3}}{z+1}\rightarrow\dfrac{\left(\left(z+1\right)^{2}\right)^{5/3}}{z+1}\rightarrow\left(z+1\right)\left(\left(z+1\right)^{2}\right)^{2/3}.\]
Factored over the integers or square-free factored form is a good
choice for radicands for other reasons too, and these forms are canonical
when the radicand is a rational expression. However,\[
\left(z^{2}+2z+1\right)\left(\left(z+1\right)^{2}\right)^{5/3}\rightarrow\left(z+1\right)^{2}\left(\left(z+1\right)^{2}\right)^{5/3}\rightarrow\left(\left(z+1\right)^{2}\right)^{8/3},\]
would require another rule that factors the \textsl{cofactor} of a
power of a power of a sum. Then, perhaps we would want another rule
to factor \textsl{sums containing} such radicands so that\[
z^{2}\left(\left(z+1\right)^{2}\right)^{5/3}+2z\left(\left(z+1\right)^{2}\right)^{5/3}+\left(\left(z+1\right)^{2}\right)^{5/3}\rightarrow\left(\left(z+1\right)^{2}\right)^{8/3}.\]
It is impossible to implement equivalence recognition for all possible
expressions $w$ representable in general purpose systems, but it
is worth expending a modest amount of execution time for default simplification
and more time for optional transformations.

The Appendix leaves most such opportunities unimplemented because
the simplifications described here are so fundamental and low level
that they should be part of the built-in transformations. Good simplification
of nested power products is more appropriately built into a system
rather than provided as an optionally loaded package that most users
are unlikely to know about and load into every session. So rather
than implementing a comprehensive package for one system, the intent
of this article is to inspire implementers of all systems to improve
some very fundamental transformations -- at least to the extent that
it can be done economically.

\subsection{Non numeric exponents}

Although not implemented in the rules of the Appendix, more generally
the exponents for forms 1 through 4 can be Gaussian fractions or even
symbolic, in which case we can still apply these transformations to
the rational numeric parts of the exponents. For example,\[
w^{3\xi+\rho}\left(w^{\xi}\right)^{3/2+\omega\pi i}\rightarrow\left(w^{\xi}\right)^{3}w^{\rho}\left(w^{\xi}\right)^{1+1/2+\omega\pi i}\rightarrow\left(w^{\xi}\right)^{4}w^{\rho}\left(w^{\xi}\right)^{1/2+\omega\pi i}\rightarrow w^{4\xi+\rho}\left(w^{\xi}\right)^{1/2+\omega\pi i}.\]

As another example, if a user has declared the variable $n$ to be
integer, then\[
w^{-n}\left(w^{2}\right)^{n+1/2}\rightarrow w^{n}\left(w^{2}\right)^{1/2}.\]

To some extent, the methods can also be extended to handle floating-point
and symbolic real expressions for exponents $\alpha$ and $\beta_{k}$.
For example,\begin{eqnarray*}
w^{4.321}\left(w^{1.234}\right)^{3/2} & \rightarrow & w^{5.555}\sqrt{w^{1.234}},\\
w^{2-\pi}\left(w^{\pi}\right)^{3/2} & \rightarrow & w^{2}\sqrt{w^{\pi}}.\end{eqnarray*}

\section{\label{sec:Summary}Summary}

This article:
\begin{enumerate}
\item shows that many widely-used computer algebra systems have significant
room for improvement at simplifying sub-expressions of the form $w^{\alpha}\left(w^{\beta_{1}}\right)^{\gamma_{1}}\cdots\left(w^{\beta_{n}}\right)^{\gamma_{n}}$;
\item defines four different simplified forms with good properties;
\item explains how to compute these forms;
\item includes a demonstration implementation of form 3 via \textsl{Mathematica}
rewrite rules.
\end{enumerate}

\section*{Acknowledgment}

I thank Sam Blake for his helpful assistance with \textsl{Mathematica},
Daniel Lichtblau for information about the algorithm in $\mathtt{Reduce}\left[\ldots\right]$,
and a referee for many fine suggestions.

\section*{Appendix: \textmd{\textsl{Mathematica}} rewrite rules for $w^{\alpha}\left(w^{\beta_{1}}\right)^{\gamma_{1}}\cdots\left(w^{\beta_{n}}\right)^{\gamma_{n}}$}

\begin{verbatim}
   (* EXTRA SIMPLIFICATION DONE BEFORE ORDINARY EVALUATION: *)

PreProductOfPowersOfPowers [(w_Plus)^(g_Rational /; !IntegerQ[g])] :=
  Block[{squareFree = FactorSquareFree[w]},
        squareFree^g /; Head[squareFree] =!= Plus];

PreProductOfPowersOfPowers [(w_^b_)^(g_Rational /; g <= -1 || g >= 1)] :=
   w^(IntegerPart[g]*b) * (w^b)^FractionalPart[g];

PreProductOfPowersOfPowers [(w_^b_)^(g_Rational /; g<=-1 || g>=1) * w_^a_. * u_]:=
  PreProductOfPowersOfPowers[w^(a+IntegerPart[g]*b) * (w^b)^FractionalPart[g] * u];

PreProductOfPowersOfPowers [(w_^b_)^g_ * w_^a_. * u_. /; Sign[a] != Sign[b*g] &&
    (Sign [a+b*Sign[g]] == Sign [b*(g-Sign[g])] ||
     Min [Abs[a], Abs[b*g]] > Min [Abs [a+b*Sign[g]], Abs [b*(g-Sign[g])]] ||
     g == -1/2 && Min [Abs[a], Abs[b/2]] == Min [Abs[a-b], Abs[b/2]])] :=
   w^(a+b*Sign[g]) * (w^b)^(g-Sign[g]) * u;

PreProductOfPowersOfPowers [f_[args__]] :=
   Apply [f, Map [PreProductOfPowersOfPowers, {args}]];

PreProductOfPowersOfPowers [anythingElse_] := anythingElse;

   (* EXTRA SIMPLIFICATION DURING ORDINARY EVALUATION: *)

Unprotect [Times];

(w_^b_)^g_ * w_^a_. * u_. /; Sign[a] != Sign[b*g] &&
    (Sign [a+b*Sign[g]] == Sign [b*(g-Sign[g])] ||
     Min [Abs[a], Abs[b*g]] > Min [Abs [a+b*Sign[g]], Abs [b*(g-Sign[g])]] ||
     Min [Abs[a], Abs[b*g]] == Min [Abs [a+b*Sign[g]], Abs [b*(g-Sign[g])]] &&
       Abs[g] > Abs [g-Sign[g]] ||
     g == -1/2 && Min [Abs[a], Abs[b/2]] == Min [Abs[a-b], Abs[b/2]]) :=
   w^(a+b*Sign[g]) * (w^b)^(g-Sign[g]) * u;

(w_^b1_)^g1_ * (w_^b2_)^g2_ * u_. /; Sign[b1*g1] != Sign[b2*g2] &&
     Abs[b2] > Abs[b1] && Sign [b2*(g2-Sign[g2])] == Sign [b1*g1 + b2*Sign[g2]] :=
   (w^b2)^(g2-Sign[g2]) * ((w^(b2*Sign[g2]) * (w^b1)^g1) * u);

Protect [Times];

   (* EXTRA SIMPLIFICATION DONE AFTER ORDINARY EVALUATION: *)

PostProductOfPowersOfPowers [w_^a_ * (w_^b_)^g_ * u_. /; IntegerQ [a/b]] :=
    PostProductOfPowersOfPowers [(w^b)^(g+a/b) * u];

PostProductOfPowersOfPowers [w_^a_.*(w_^b_)^(g_Rational /; g<=-1 || g>=1)*u_.
     /; !IntegerQ [(a + b*IntegerPart[g])/b]] :=
  PostProductOfPowersOfPowers[(u*w^(a+b*IntegerPart[b]))*(w^b)^FractionalPart[g]];

PostProductOfPowersOfPowers [f_[args__]] :=
   Apply [f, Map [PostProductOfPowersOfPowers, {args}]];

PostProductOfPowersOfPowers [anythingElse_] := anythingElse;

$Post = PostProductOfPowersOfPowers;    $Pre = PreProductOfPowersOfPowers;
\end{verbatim}

\section*{Tables \ref{Flo:AlgorithmAtable} through \ref{Flo:Form1MinusForm4Table}}

\begin{table}[h]
\caption{\textbf{Unflawed} results \textbf{of Appendix rewrite rules} for 1st
row $\times$ 1st column.\protect \\
Compare with Tables \ref{Flo:MathematicaDefaultTable} through
\ref{Flo:MaximaFullratsimpTableAndMapleSimplifyTable}}

\begin{centering}
\label{Flo:AlgorithmAtable}
\par\end{centering}

\centering{}
\begin{tabular}{|c||c|c|c|c|c|c|}
\hline 
$\downarrow\!\!\!\overrightarrow{\times}$ & \multicolumn{1}{c|}{$\left(w^{2}\right)^{-7/3}$} & \multicolumn{1}{c|}{$\left(w^{2}\right)^{-4/3}$} & \multicolumn{1}{c|}{$\left(w^{2}\right)^{-1/3}$} & \multicolumn{1}{c|}{$\left(w^{2}\right)^{2/3}$} & \multicolumn{1}{c|}{$\left(w^{2}\right)^{5/3}$} & \multicolumn{1}{c|}{$\left(w^{2}\right)^{8/3}$}\tabularnewline
\hline
\hline 
$w^{-3}$ & $\frac{1}{w^{7}\left(w^{2}\right)^{1/3}}$ & $\frac{1}{w^{5}\left(w^{2}\right)^{1/3}}$ & $\frac{1}{w^{3}\left(w^{2}\right)^{1/3}}$ & $\frac{1}{w\left(w^{2}\right)^{1/3}}$ & $\frac{w}{\left(w^{2}\right)^{1/3}}$ & $w\left(w^{2}\right)^{2/3}$\tabularnewline
\hline 
$w^{-2}$ & $\frac{1}{\left(w^{2}\right)^{10/3}}$ & $\frac{1}{\left(w^{2}\right)^{7/3}}$ & $\frac{1}{\left(w^{2}\right)^{4/3}}$ & $\frac{1}{\left(w^{2}\right)^{1/3}}$ & $\left(w^{2}\right)^{2/3}$ & $\left(w^{2}\right)^{5/3}$\tabularnewline
\hline 
$w^{-1}$ & $\frac{1}{w^{5}\left(w^{2}\right)^{1/3}}$ & $\frac{1}{w^{3}\left(w^{2}\right)^{1/3}}$ & $\frac{1}{w\left(w^{2}\right)^{1/3}}$ & $\frac{w}{\left(w^{2}\right)^{1/3}}$ & $w\left(w^{2}\right)^{2/3}$ & $w^{3}\left(w^{2}\right)^{2/3}$\tabularnewline
\hline 
$w^{0}$ & $\frac{1}{\left(w^{2}\right)^{7/3}}$ & $\frac{1}{\left(w^{2}\right)^{4/3}}$ & $\frac{1}{\left(w^{2}\right)^{1/3}}$ & $\left(w^{2}\right)^{2/3}$ & $\left(w^{2}\right)^{5/3}$ & $\left(w^{2}\right)^{8/3}$\tabularnewline
\hline 
$w^{1}$ & $\frac{1}{w^{3}\left(w^{2}\right)^{1/3}}$ & $\frac{1}{w\left(w^{2}\right)^{1/3}}$ & $\frac{w}{\left(w^{2}\right)^{1/3}}$ & $w\left(w^{2}\right)^{2/3}$ & $w^{3}\left(w^{2}\right)^{2/3}$ & $w^{5}\left(w^{2}\right)^{2/3}$\tabularnewline
\hline 
$w^{2}$ & $\frac{1}{\left(w^{2}\right)^{4/3}}$ & $\frac{1}{\left(w^{2}\right)^{1/3}}$ & $\left(w^{2}\right)^{2/3}$ & $\left(w^{2}\right)^{5/3}$ & $\left(w^{2}\right)^{8/3}$ & $\left(w^{2}\right)^{11/3}$\tabularnewline
\hline
\end{tabular}
\end{table}

\begin{table}[h]
\begin{centering}
\caption{\textbf{\textsl{Mathematica}}\textbf{ 8 default} simplification for
1st row $\times$ 1st column, with flaw numbers. Compare with Table
\ref{Flo:AlgorithmAtable}.}
\label{Flo:MathematicaDefaultTable}
\par\end{centering}

\centering{}\begin{tabular}{|c||c|c|c|c|c|c|c|c|c|c|c|c|}
\hline 
$\downarrow\!\!\!\overrightarrow{\times}$ & \multicolumn{2}{c|}{$\left(w^{2}\right)^{-7/3}$} & \multicolumn{2}{c|}{$\left(w^{2}\right)^{-4/3}$} & \multicolumn{2}{c|}{$\left(w^{2}\right)^{-1/3}$} & \multicolumn{2}{c|}{$\left(w^{2}\right)^{2/3}$} & \multicolumn{2}{c|}{$\left(w^{2}\right)^{5/3}$} & \multicolumn{2}{c|}{$\left(w^{2}\right)^{8/3}$}\tabularnewline
\hline
\end{tabular}
\end{table}
\begin{table}[H]
\begin{centering}
\caption{\textbf{\textsl{Mathematica}}\textbf{ 8 FullSimplify{[}...{]}} for
1st row $\times$ 1st column, with flaw numbers. Compare with Table
\ref{Flo:AlgorithmAtable}.}
\label{Flo:MathematicaFullSimplifyTable}
\par\end{centering}

\centering{}\begin{tabular}{|c||c|c|c|c|c|c|c|c|c|c|c|c|}
\hline 
$\downarrow\!\!\!\overrightarrow{\times}$ & \multicolumn{2}{c|}{$\left(w^{2}\right)^{-7/3}$} & \multicolumn{2}{c|}{$\left(w^{2}\right)^{-4/3}$} & \multicolumn{2}{c|}{$\left(w^{2}\right)^{-1/3}$} & \multicolumn{2}{c|}{$\left(w^{2}\right)^{2/3}$} & \multicolumn{2}{c|}{$\left(w^{2}\right)^{5/3}$} & \multicolumn{2}{c|}{$\left(w^{2}\right)^{8/3}$}\tabularnewline
\hline
\hline 
$w^{-3}$ & $\frac{\left(w^{2}\right)^{2/3}}{w^{9}}$ & {\tiny \negthinspace{}\negthinspace{}\negthinspace{}\negthinspace{}\negthinspace{}\negthinspace{}\negthinspace{}\negthinspace{}$\begin{array}{c}
\boldsymbol{\mathbf{1}}\\
3\\
5\end{array}$\negthinspace{}\negthinspace{}\negthinspace{}\negthinspace{}\negthinspace{}\negthinspace{}\negthinspace{}\negthinspace{}} & $\frac{\left(w^{2}\right)^{2/3}}{w^{7}}$ & {\tiny \negthinspace{}\negthinspace{}\negthinspace{}\negthinspace{}\negthinspace{}\negthinspace{}\negthinspace{}\negthinspace{}$\begin{array}{c}
\boldsymbol{\mathbf{1}}\\
3\\
5\end{array}$\negthinspace{}\negthinspace{}\negthinspace{}\negthinspace{}\negthinspace{}\negthinspace{}\negthinspace{}\negthinspace{}} & $\frac{\left(w^{2}\right)^{2/3}}{w^{5}}$ & {\tiny \negthinspace{}\negthinspace{}\negthinspace{}\negthinspace{}\negthinspace{}\negthinspace{}\negthinspace{}\negthinspace{}$\begin{array}{c}
\boldsymbol{\mathbf{1}}\\
3\\
5\end{array}$\negthinspace{}\negthinspace{}\negthinspace{}\negthinspace{}\negthinspace{}\negthinspace{}\negthinspace{}\negthinspace{}} & $\frac{\left(w^{2}\right)^{2/3}}{w^{3}}$ & {\tiny \negthinspace{}\negthinspace{}\negthinspace{}\negthinspace{}\negthinspace{}\negthinspace{}\negthinspace{}\negthinspace{}$\begin{array}{c}
3\\
5\end{array}$\negthinspace{}\negthinspace{}\negthinspace{}\negthinspace{}\negthinspace{}\negthinspace{}\negthinspace{}\negthinspace{}} & $\frac{w}{\left(w^{2}\right)^{1/3}}$ &  & $w\left(w^{2}\right)^{2/3}$ & \tabularnewline
\hline
\end{tabular}
\end{table}

\begin{table}[h]
\begin{centering}
\caption{\textbf{\textsl{Derive}}\textbf{ 6 default} simplify for 1st row $\times$
1st column, with flaw numbers.\protect \\
Compare with Table \ref{Flo:AlgorithmAtable}.}
\label{Flo:DeriveDefaultTable}
\par\end{centering}

\centering{}\begin{tabular}{|c||c|c|c|c|c|c|c|c|c|c|c|c|}
\hline 
$\downarrow\!\!\!\overrightarrow{\times}$ & \multicolumn{2}{c|}{$\left(w^{2}\right)^{-7/3}$} & \multicolumn{2}{c|}{$\left(w^{2}\right)^{-4/3}$} & \multicolumn{2}{c|}{$\left(w^{2}\right)^{-1/3}$} & \multicolumn{2}{c|}{$\left(w^{2}\right)^{2/3}$} & \multicolumn{2}{c|}{$\left(w^{2}\right)^{5/3}$} & \multicolumn{2}{c|}{$\left(w^{2}\right)^{8/3}$}\tabularnewline
\hline
\end{tabular}
\end{table}

\begin{table}[h]
\begin{centering}
\caption{\textbf{TI-CAS 3.1 default} simplify for 1st row $\times$ 1st column,
with flaw numbers.\protect \\
Compare with Table \ref{Flo:AlgorithmAtable}.}
\label{Flo:TICASDefaultTable}
\par\end{centering}

\centering{}\begin{tabular}{|c||c|c|c|c|c|c|c|c|c|c|c|c|}
\hline 
$\downarrow\!\!\!\overrightarrow{\times}$ & \multicolumn{2}{c|}{$\left(w^{2}\right)^{-7/3}$} & \multicolumn{2}{c|}{$\left(w^{2}\right)^{-4/3}$} & \multicolumn{2}{c|}{$\left(w^{2}\right)^{-1/3}$} & \multicolumn{2}{c|}{$\left(w^{2}\right)^{2/3}$} & \multicolumn{2}{c|}{$\left(w^{2}\right)^{5/3}$} & \multicolumn{2}{c|}{$\left(w^{2}\right)^{8/3}$}\tabularnewline
\hline
\hline 
$w^{-3}$ & $\frac{1}{w^{7}\left(w^{2}\right)^{1/3}}$ &  & $\frac{1}{w^{5}\left(w^{2}\right)^{1/3}}$ &  & $\frac{1}{w^{3}\left(w^{2}\right)^{1/3}}$ & {\tiny \negthinspace{}\negthinspace{}\negthinspace{}\negthinspace{}\negthinspace{}\negthinspace{}\negthinspace{}\negthinspace{}$\begin{array}{c}
\mathbf{2}\end{array}$\negthinspace{}\negthinspace{}\negthinspace{}\negthinspace{}\negthinspace{}\negthinspace{}\negthinspace{}\negthinspace{}} & $\frac{\left(w^{2}\right)^{2/3}}{w^{3}}$ & {\tiny \negthinspace{}\negthinspace{}\negthinspace{}\negthinspace{}\negthinspace{}\negthinspace{}\negthinspace{}\negthinspace{}$\begin{array}{c}
\mathbf{2}\\
3\\
5\\
6\end{array}$\negthinspace{}\negthinspace{}\negthinspace{}\negthinspace{}\negthinspace{}\negthinspace{}\negthinspace{}\negthinspace{}} & $\frac{\left(w^{2}\right)^{2/3}}{w}$ & {\tiny \negthinspace{}\negthinspace{}\negthinspace{}\negthinspace{}\negthinspace{}\negthinspace{}\negthinspace{}\negthinspace{}$\begin{array}{c}
3\\
5\\
6\end{array}$\negthinspace{}\negthinspace{}\negthinspace{}\negthinspace{}\negthinspace{}\negthinspace{}\negthinspace{}\negthinspace{}} & $w\left(w^{2}\right)^{2/3}$ & \tabularnewline
\hline 
$w^{-2}$ & $\frac{1}{\left(w^{2}\right)^{10/3}}$ &  & $\frac{1}{\left(w^{2}\right)^{7/3}}$ &  & $\frac{1}{\left(w^{2}\right)^{4/3}}$ &  & $\frac{1}{\left(w^{2}\right)^{1/3}}$ &  & $\left(w^{2}\right)^{2/3}$ &  & $\left(w^{2}\right)^{5/3}$ & \tabularnewline
\hline
\end{tabular}
\end{table}
\begin{table}[h]
\begin{centering}
\caption{\textbf{Maple 15 and Maxima 5.24 default} simplification for 1st row
$\times$ 1st column, with flaw numbers. Compare with Table \ref{Flo:AlgorithmAtable}.}
\label{Flo:MaximaAndMapleDefaultTable}
\par\end{centering}

\centering{}\begin{tabular}{|c||c|c|c|c|c|c|c|c|c|c|c|c|}
\hline 
$\downarrow\!\!\!\overrightarrow{\times}$ & \multicolumn{2}{c|}{$\left(w^{2}\right)^{-7/3}$} & \multicolumn{2}{c|}{$\left(w^{2}\right)^{-4/3}$} & \multicolumn{2}{c|}{$\left(w^{2}\right)^{-1/3}$} & \multicolumn{2}{c|}{$\left(w^{2}\right)^{2/3}$} & \multicolumn{2}{c|}{$\left(w^{2}\right)^{5/3}$} & \multicolumn{2}{c|}{$\left(w^{2}\right)^{8/3}$}\tabularnewline
\hline
\end{tabular}
\end{table}

\begin{table}[h]
\begin{centering}
\caption{\textbf{Maxima 5.24 fullratsimp(...)} and \textbf{Maple 15 simplify(...)}
for 1st row $\times$ 1st column, with flaw numbers. Compare with
Table \ref{Flo:AlgorithmAtable}.}
\label{Flo:MaximaFullratsimpTableAndMapleSimplifyTable}
\par\end{centering}

\centering{}\begin{tabular}{|c||c|c|c|c|c|c|c|c|c|c|c|c|}
\hline 
$\downarrow\!\!\!\overrightarrow{\times}$ & \multicolumn{2}{c|}{$\left(w^{2}\right)^{-7/3}$} & \multicolumn{2}{c|}{$\left(w^{2}\right)^{-4/3}$} & \multicolumn{2}{c|}{$\left(w^{2}\right)^{-1/3}$} & \multicolumn{2}{c|}{$\left(w^{2}\right)^{2/3}$} & \multicolumn{2}{c|}{$\left(w^{2}\right)^{5/3}$} & \multicolumn{2}{c|}{$\left(w^{2}\right)^{8/3}$}\tabularnewline
\hline
\hline 
$w^{-3}$ & $\frac{1}{w^{7}\left(w^{2}\right)^{1/3}}$ &  & $\frac{1}{w^{5}\left(w^{2}\right)^{1/3}}$ &  & $\frac{1}{w^{3}\left(w^{2}\right)^{1/3}}$ &  & $\frac{\left(w^{2}\right)^{2/3}}{w^{3}}$ & {\tiny \negthinspace{}\negthinspace{}\negthinspace{}\negthinspace{}\negthinspace{}\negthinspace{}\negthinspace{}\negthinspace{}$\begin{array}{c}
3\\
5\\
6\end{array}$\negthinspace{}\negthinspace{}\negthinspace{}\negthinspace{}\negthinspace{}\negthinspace{}\negthinspace{}\negthinspace{}} & $\frac{\left(w^{2}\right)^{2/3}}{w}$ & {\tiny \negthinspace{}\negthinspace{}\negthinspace{}\negthinspace{}\negthinspace{}\negthinspace{}\negthinspace{}\negthinspace{}$\begin{array}{c}
3\\
5\\
6\end{array}$\negthinspace{}\negthinspace{}\negthinspace{}\negthinspace{}\negthinspace{}\negthinspace{}\negthinspace{}\negthinspace{}} & $w\left(w^{2}\right)^{2/3}$ & \tabularnewline
\hline
\end{tabular}
\end{table}

\begin{table}[h]
\begin{centering}
\caption{\textbf{Unflawed }results of \textbf{Appendix rewrite rules} for 1st
row \texttimes{} 1st column.\protect \\
Compare with Tables \ref{Flo:MathematicaDefaultSqrtReciprocalTable}
through \ref{Flo:MapleSimplifySqrtReciprocal}.}
\label{Flo:AlgorithmASqrtReciprocalTable}
\par\end{centering}

\centering{}\begin{tabular}{|c||c|c|c|c|c|c|}
\hline 
$\downarrow\!\!\!\overrightarrow{\times}$ & \multicolumn{1}{c|}{$\left(w^{-2}\right)^{-5/2}$} & \multicolumn{1}{c|}{$\left(w^{-2}\right)^{-3/2}$} & \multicolumn{1}{c|}{$\left(w^{-2}\right)^{-1/2}$} & \multicolumn{1}{c|}{$\left(w^{-2}\right)^{1/2}$} & \multicolumn{1}{c|}{$\left(w^{-2}\right)^{3/2}$} & \multicolumn{1}{c|}{$\left(w^{-2}\right)^{5/2}$}\tabularnewline
\hline
\hline 
$w^{-3}$ & \textsl{$\frac{w}{\sqrt{\frac{1}{w^{2}}}}$} & $\sqrt{\frac{1}{w^{2}}}w$ & $\frac{\sqrt{\frac{1}{w^{2}}}}{w}$ & $\frac{\sqrt{\frac{1}{w^{2}}}}{w^{3}}$ & $\frac{\sqrt{\frac{1}{w^{2}}}}{w^{5}}$ & $\frac{\sqrt{\frac{1}{w^{2}}}}{w^{7}}$\tabularnewline
\hline 
$w^{-2}$ & $\frac{1}{\left(\frac{1}{w^{2}}\right)^{3/2}}$ & $\frac{1}{\sqrt{\frac{1}{w^{2}}}}$ & $\sqrt{\frac{1}{w^{2}}}$ & $\left(\frac{1}{w^{2}}\right)^{3/2}$ & $\left(\frac{1}{w^{2}}\right)^{5/2}$ & $\left(\frac{1}{w^{2}}\right)^{7/2}$\tabularnewline
\hline 
$w^{-1}$ & $\frac{w^{3}}{\sqrt{\frac{1}{w^{2}}}}$ & $\frac{w}{\sqrt{\frac{1}{w^{2}}}}$ & $\sqrt{\frac{1}{w^{2}}}w$ & $\frac{\sqrt{\frac{1}{w^{2}}}}{w}$ & $\frac{\sqrt{\frac{1}{w^{2}}}}{w^{3}}$ & $\frac{\sqrt{\frac{1}{w^{2}}}}{w^{5}}$\tabularnewline
\hline 
$w^{0}$ & $\frac{1}{\left(\frac{1}{w^{2}}\right)^{5/2}}$ & $\frac{1}{\left(\frac{1}{w^{2}}\right)^{3/2}}$ & $\frac{1}{\sqrt{\frac{1}{w^{2}}}}$ & $\sqrt{\frac{1}{w^{2}}}$ & $\left(\frac{1}{w^{2}}\right)^{3/2}$ & $\left(\frac{1}{w^{2}}\right)^{5/2}$\tabularnewline
\hline 
$w^{1}$ & $\frac{w^{5}}{\sqrt{\frac{1}{w^{2}}}}$ & $\frac{w^{3}}{\sqrt{\frac{1}{w^{2}}}}$ & $\frac{w}{\sqrt{\frac{1}{w^{2}}}}$ & $\sqrt{\frac{1}{w^{2}}}w$ & $\frac{\sqrt{\frac{1}{w^{2}}}}{w}$ & $\frac{\sqrt{\frac{1}{w^{2}}}}{w^{3}}$\tabularnewline
\hline 
$w^{2}$ & $\frac{1}{\left(\frac{1}{w^{2}}\right)^{7/2}}$ & $\frac{1}{\left(\frac{1}{w^{2}}\right)^{5/2}}$ & $\frac{1}{\left(\frac{1}{w^{2}}\right)^{3/2}}$ & $\frac{1}{\sqrt{\frac{1}{w^{2}}}}$ & $\sqrt{\frac{1}{w^{2}}}$ & $\left(\frac{1}{w^{2}}\right)^{3/2}$\tabularnewline
\hline
\end{tabular}
\end{table}
\begin{table}[h]
\begin{centering}
\caption{\textbf{\textsl{Mathematica}}\textbf{ 8 default} simplify for 1st
row \texttimes{} 1st column, with flaw numbers. Compare with Table
\ref{Flo:AlgorithmASqrtReciprocalTable}.}
\label{Flo:MathematicaDefaultSqrtReciprocalTable}
\par\end{centering}

\centering{}\begin{tabular}{|c||c|c|c|c|c|c|c|c|c|c|c|c|}
\hline 
$\downarrow\!\!\!\overrightarrow{\times}$ & \multicolumn{2}{c|}{$\left(w^{-2}\right)^{-5/2}$} & \multicolumn{2}{c|}{$\left(w^{-2}\right)^{-3/2}$} & \multicolumn{2}{c|}{$\left(w^{-2}\right)^{-1/2}$} & \multicolumn{2}{c|}{$\left(w^{-2}\right)^{1/2}$} & \multicolumn{2}{c|}{$\left(w^{-2}\right)^{3/2}$} & \multicolumn{2}{c|}{$\left(w^{-2}\right)^{5/2}$}\tabularnewline
\hline
\end{tabular}
\end{table}

\begin{center}
\begin{table}[h]
\begin{centering}
\caption{\textbf{\textsl{Mathematica}}\textbf{ 8 FullSimplify{[}...{]}} for
1st row \texttimes{} 1st column, with flaw numbers. Compare with Table
\ref{Flo:AlgorithmASqrtReciprocalTable}.}
\label{Flo:MathematicaFullSimplifySqrtReciprocalTable}
\par\end{centering}

\centering{}\begin{tabular}{|c||c|c|c|c|c|c|c|c|c|c|c|c|}
\hline 
$\downarrow\!\!\!\overrightarrow{\times}$ & \multicolumn{2}{c|}{$\left(w^{-2}\right)^{-5/2}$} & \multicolumn{2}{c|}{$\left(w^{-2}\right)^{-3/2}$} & \multicolumn{2}{c|}{$\left(w^{-2}\right)^{-1/2}$} & \multicolumn{2}{c|}{$\left(w^{-2}\right)^{1/2}$} & \multicolumn{2}{c|}{$\left(w^{-2}\right)^{3/2}$} & \multicolumn{2}{c|}{$\left(w^{-2}\right)^{5/2}$}\tabularnewline
\hline
\hline 
$w^{-3}$ & $\sqrt{\frac{1}{w^{2}}}w^{3}$ & {\tiny \negthinspace{}\negthinspace{}\negthinspace{}\negthinspace{}\negthinspace{}\negthinspace{}\negthinspace{}\negthinspace{}$\begin{array}{c}
3\\
6\end{array}$\negthinspace{}\negthinspace{}\negthinspace{}\negthinspace{}\negthinspace{}\negthinspace{}\negthinspace{}\negthinspace{}} & $\sqrt{\frac{1}{w^{2}}}w$ &  & $\frac{\sqrt{\frac{1}{w^{2}}}}{w}$ &  & $\frac{\sqrt{\frac{1}{w^{2}}}}{w^{3}}$ &  & \textsl{$\frac{\left(\frac{1}{w^{2}}\right)^{3/2}}{w^{3}}$} & {\tiny \negthinspace{}\negthinspace{}\negthinspace{}\negthinspace{}\negthinspace{}\negthinspace{}\negthinspace{}\negthinspace{}$\begin{array}{c}
5\\
6\end{array}$\negthinspace{}\negthinspace{}\negthinspace{}\negthinspace{}\negthinspace{}\negthinspace{}\negthinspace{}\negthinspace{}} & \textsl{$\frac{\left(\frac{1}{w^{2}}\right)^{5/2}}{w^{3}}$} & {\tiny \negthinspace{}\negthinspace{}\negthinspace{}\negthinspace{}\negthinspace{}\negthinspace{}\negthinspace{}\negthinspace{}$\begin{array}{c}
5\\
6\end{array}$\negthinspace{}\negthinspace{}\negthinspace{}\negthinspace{}\negthinspace{}\negthinspace{}\negthinspace{}\negthinspace{}}\tabularnewline
\hline 
$w^{-2}$ & $\frac{1}{\left(\frac{1}{w^{2}}\right)^{3/2}}$ &  & $\frac{1}{\sqrt{\frac{1}{w^{2}}}}$ &  & $\sqrt{\frac{1}{w^{2}}}$ &  & $\left(\frac{1}{w^{2}}\right)^{3/2}$ &  & $\left(\frac{1}{w^{2}}\right)^{5/2}$ &  & $\left(\frac{1}{w^{2}}\right)^{7/2}$ & \tabularnewline
\hline
\end{tabular}
\end{table}

\par\end{center}

\begin{center}
\begin{table}[h]
\begin{centering}
\caption{\textbf{\textsl{Derive}}\textbf{ 6 default} simplify for 1st row \texttimes{}
1st column, with flaw numbers.\protect \\
Compare with Table \ref{Flo:AlgorithmASqrtReciprocalTable}.}
\label{Flo:DeriveSqrtReciprocalTable}
\par\end{centering}

\centering{}\begin{tabular}{|c||c|c|c|c|c|c|c|c|c|c|c|c|}
\hline 
$\downarrow\!\!\!\overrightarrow{\times}$ & \multicolumn{2}{c|}{$\left(w^{-2}\right)^{-5/2}$} & \multicolumn{2}{c|}{$\left(w^{-2}\right)^{-3/2}$} & \multicolumn{2}{c|}{$\left(w^{-2}\right)^{-1/2}$} & \multicolumn{2}{c|}{$\left(w^{-2}\right)^{1/2}$} & \multicolumn{2}{c|}{$\left(w^{-2}\right)^{3/2}$} & \multicolumn{2}{c|}{$\left(w^{-2}\right)^{5/2}$}\tabularnewline
\hline
\hline 
$w^{-3}$ & $w^{3}\sqrt{\frac{1}{w^{2}}}$ & {\tiny \negthinspace{}\negthinspace{}\negthinspace{}\negthinspace{}\negthinspace{}\negthinspace{}\negthinspace{}\negthinspace{}$\begin{array}{c}
3\end{array}$\negthinspace{}\negthinspace{}\negthinspace{}\negthinspace{}\negthinspace{}\negthinspace{}\negthinspace{}\negthinspace{}} & $\sqrt{\frac{1}{w^{2}}}w$ &  & $\frac{\sqrt{\frac{1}{w^{2}}}}{w}$ &  & $\frac{\sqrt{\frac{1}{w^{2}}}}{w^{3}}$ &  & $\frac{\sqrt{\frac{1}{w^{2}}}}{w^{5}}$ &  & $\frac{\sqrt{\frac{1}{w^{2}}}}{w^{7}}$ & \tabularnewline
\hline 
$w^{-2}$ & $\frac{1}{\left(\frac{1}{w^{2}}\right)^{3/2}}$ &  & $\frac{1}{\sqrt{\frac{1}{w^{2}}}}$ &  & $\sqrt{\frac{1}{w^{2}}}$ &  & $\left(\frac{1}{w^{2}}\right)^{3/2}$ &  & $\left(\frac{1}{w^{2}}\right)^{5/2}$ &  & $\left(\frac{1}{w^{2}}\right)^{7/2}$ & \tabularnewline
\hline
\end{tabular}
\end{table}

\par\end{center}

\begin{center}
\begin{table}[h]
\begin{centering}
\caption{\textbf{TI-CAS and Maple default simplification} for 1st row \texttimes{}
1st column, with flaw numbers. Compare with Table \ref{Flo:AlgorithmASqrtReciprocalTable}.}
\label{Flo:TIandMapleDefaultSqrtReciprocal}
\par\end{centering}

\centering{}\begin{tabular}{|c||c|c|c|c|c|c|c|c|c|c|c|c|}
\hline 
$\downarrow\!\!\!\overrightarrow{\times}$ & \multicolumn{2}{c|}{$\left(w^{-2}\right)^{-5/2}$} & \multicolumn{2}{c|}{$\left(w^{-2}\right)^{-3/2}$} & \multicolumn{2}{c|}{$\left(w^{-2}\right)^{-1/2}$} & \multicolumn{2}{c|}{$\left(w^{-2}\right)^{1/2}$} & \multicolumn{2}{c|}{$\left(w^{-2}\right)^{3/2}$} & \multicolumn{2}{c|}{$\left(w^{-2}\right)^{5/2}$}\tabularnewline
\hline
\end{tabular}
\end{table}

\par\end{center}

\begin{table}[h]

\caption{\textbf{Maxima 5.24 default} simplify for 1st row \texttimes{} 1st
column, with flaw numbers. Compare with Table \ref{Flo:AlgorithmASqrtReciprocalTable}.}
\label{Flo:MaximaDefaultSqrtReciprocal}

\centering{}\begin{tabular}{|c||c|c|c|c|c|c|c|c|c|c|c|c|}
\hline 
$\downarrow\!\!\!\overrightarrow{\times}$ & \multicolumn{2}{c|}{$\left(w^{-2}\right)^{-5/2}$} & \multicolumn{2}{c|}{$\left(w^{-2}\right)^{-3/2}$} & \multicolumn{2}{c|}{$\left(w^{-2}\right)^{-1/2}$} & \multicolumn{2}{c|}{$\left(w^{-2}\right)^{1/2}$} & \multicolumn{2}{c|}{$\left(w^{-2}\right)^{3/2}$} & \multicolumn{2}{c|}{$\left(w^{-2}\right)^{5/2}$}\tabularnewline
\hline
\end{tabular}
\end{table}

\begin{center}
\begin{table}[h]
\caption{\textbf{Maxima 5.24 fullratsimp(...)} for 1st row \texttimes{} 1st
column, with flaw numbers. Compare with Table \ref{Flo:AlgorithmASqrtReciprocalTable}.}
\label{Flo:MaximaFullratsimpSqrtRecip}

\centering{}\begin{tabular}{|c||c|c|c|c|c|c|c|c|c|c|c|c|}
\hline 
$\downarrow\!\!\!\overrightarrow{\times}$ & \multicolumn{2}{c|}{$\left(w^{-2}\right)^{-5/2}$} & \multicolumn{2}{c|}{$\left(w^{-2}\right)^{-3/2}$} & \multicolumn{2}{c|}{$\left(w^{-2}\right)^{-1/2}$} & \multicolumn{2}{c|}{$\left(w^{-2}\right)^{1/2}$} & \multicolumn{2}{c|}{$\left(w^{-2}\right)^{3/2}$} & \multicolumn{2}{c|}{$\left(w^{-2}\right)^{5/2}$}\tabularnewline
\hline
\end{tabular}
\end{table}

\par\end{center}

\begin{center}
\begin{table}[h]
\begin{centering}
\caption{\textbf{Unflawed }results of \textbf{Maple simplify(...)} for 1st
row \texttimes{} 1st column -- a variant of form 4.\protect \\
Compare with Table \ref{Flo:AlgorithmASqrtReciprocalTable}.}
\label{Flo:MapleSimplifySqrtReciprocal}
\par\end{centering}

\centering{}\begin{tabular}{|c||c|c|c|c|c|c|}
\hline 
$\downarrow\!\!\!\overrightarrow{\times}$ & \multicolumn{1}{c|}{$\left(w^{-2}\right)^{-5/2}$} & \multicolumn{1}{c|}{$\left(w^{-2}\right)^{-3/2}$} & \multicolumn{1}{c|}{$\left(w^{-2}\right)^{-1/2}$} & \multicolumn{1}{c|}{$\left(w^{-2}\right)^{1/2}$} & \multicolumn{1}{c|}{$\left(w^{-2}\right)^{3/2}$} & \multicolumn{1}{c|}{$\left(w^{-2}\right)^{5/2}$}\tabularnewline
\hline
\hline 
$w^{-3}$ & $w^{2}\mathrm{csgn}\left(\frac{1}{w}\right)$ & $\mathrm{csgn}\left(\frac{1}{w}\right)$ & $\frac{\mathrm{csgn}\left(\frac{1}{w}\right)}{w^{2}}$ & $\frac{\mathrm{csgn}\left(\frac{1}{w}\right)}{w^{4}}$ & $\frac{\mathrm{csgn}\left(\frac{1}{w}\right)}{w^{6}}$ & $\frac{\mathrm{csgn}\left(\frac{1}{w}\right)}{w^{8}}$\tabularnewline
\hline 
$w^{-2}$ & $w^{3}\mathrm{csgn}\left(\frac{1}{w}\right)$ & $w\,\mathrm{csgn}\left(\frac{1}{w}\right)$ & $\frac{\mathrm{csgn}\left(\frac{1}{w}\right)}{w}$ & $\frac{\mathrm{csgn}\left(\frac{1}{w}\right)}{w^{3}}$ & $\frac{\mathrm{csgn}\left(\frac{1}{w}\right)}{w^{5}}$ & $\frac{\mathrm{csgn}\left(\frac{1}{w}\right)}{w^{7}}$\tabularnewline
\hline 
$w^{-1}$ & $w^{4}\mathrm{csgn}\left(\frac{1}{w}\right)$ & $w^{2}\mathrm{csgn}\left(\frac{1}{w}\right)$ & $\mathrm{csgn}\left(\frac{1}{w}\right)$ & $\frac{\mathrm{csgn}\left(\frac{1}{w}\right)}{w^{2}}$ & $\frac{\mathrm{csgn}\left(\frac{1}{w}\right)}{w^{4}}$ & $\frac{\mathrm{csgn}\left(\frac{1}{w}\right)}{w^{6}}$\tabularnewline
\hline 
$w^{0}$ & $w^{5}\mathrm{csgn}\left(\frac{1}{w}\right)$ & $w^{3}\mathrm{csgn}\left(\frac{1}{w}\right)$ & $w\,\mathrm{csgn}\left(\frac{1}{w}\right)$ & $\frac{\mathrm{csgn}\left(\frac{1}{w}\right)}{w}$ & $\frac{\mathrm{csgn}\left(\frac{1}{w}\right)}{w^{3}}$ & $\frac{\mathrm{csgn}\left(\frac{1}{w}\right)}{w^{5}}$\tabularnewline
\hline
\end{tabular}
\end{table}
\begin{table}[h]
\begin{centering}
\caption{Simplification of $w/\sqrt{w^{2}}$, with flaw numbers}
\label{Flo:WonSqrtWSquaredTable}
\par\end{centering}

\centering{}\begin{tabular}{|c|c||c|c|}
\hline 
system & transformations & \multicolumn{2}{c|}{input $\frac{w}{\sqrt{w^{2}}}$}\tabularnewline
\hline
\hline 
 & Appendix rewrite rules & $\frac{\sqrt{w^{2}}}{w}$ & \tabularnewline
\hline 
\textsl{Mathematica} & default & \textsl{$\frac{w}{\sqrt{w^{2}}}$} & \textbf{2}, 6, 8\tabularnewline
\hline 
\textsl{Mathematica} & FullSimplify(...) & \textsl{$\frac{w}{\sqrt{w^{2}}}$} & 6, 8\tabularnewline
\hline 
\textsl{Derive} & default & $\frac{\sqrt{w^{2}}}{w}$ & \tabularnewline
\hline 
TI-CAS & default & \textsl{$\frac{w}{\sqrt{w^{2}}}$} & \textbf{2}, 6, 8\tabularnewline
\hline 
Maple & default & \textsl{$\frac{w}{\sqrt{w^{2}}}$} & \textbf{2}, 6, 8\tabularnewline
\hline 
Maxima & default & \textsl{$\frac{w}{\sqrt{w^{2}}}$} & \textbf{2}, 6, 8\tabularnewline
\hline 
Maxima & fullratsimp(...) & \textsl{$\frac{w}{\sqrt{w^{2}}}$} & 6, 8\tabularnewline
\hline 
Maple & simplify(...) & $\mathrm{csgn}\left(w\right)$ & \tabularnewline
\hline
\end{tabular}
\end{table}
\begin{table}[h]

\begin{centering}
\caption{Simplification of $\sqrt{w^{2}}/w-(-1)^{\left(1/2\right)\left(\arg(w^{2})-2\arg(w)\right)/\pi}$,
with flaw numbers.}
\label{Flo:Form1MinusForm4Table}\begin{tabular}{|c|c|c|}
\hline 
system and transformation & \multicolumn{2}{c|}{result and flaw numbers}\tabularnewline
\hline
\hline 
Appendix rewrite rules & $-i^{\left(-2\,\mathrm{Arg}\left[w\right]+\mathrm{Arg}\left[w^{2}\right]\right)/\pi}+\frac{\sqrt{w^{2}}}{w}$ & \textbf{2}\tabularnewline
\hline 
\textsl{Mathematica }default & $-i^{\left(-2\,\mathrm{Arg}\left[w\right]+\mathrm{Arg}\left[w^{2}\right]\right)/\pi}+\frac{\sqrt{w^{2}}}{w}$ & \textbf{2}\tabularnewline
\hline 
\textsl{Mathematica }FullSimplify(...) & $-i^{\left(-2\mathrm{\, Arg}\left[w\right]+\mathrm{Arg}\left[w^{2}\right]\right)/\pi}+\frac{\sqrt{w^{2}}}{w}$ & \textbf{2}\tabularnewline
\hline 
\textsl{Derive }default & $\frac{\sqrt{w^{2}}}{w}-\mathrm{IF}\left(w=0,\:0,\:\left(-1\right)^{\left(\mathrm{PHASE}\left(w^{2}\right)-2\mathrm{\, PHASE}\left(w\right)\right)/(2\pi)}\right)$ & \textbf{2}\tabularnewline
\hline 
TI-CAS, default & \textsl{$\frac{\sqrt{w{}^{2}}}{w}-e^{\pi i\begin{cases}
0, & w=0\\
\frac{\left(\mathrm{angle}\left(w^{2}\right)-2\,\mathrm{angle}(w)\right)1/2}{\pi} & \mathrm{else}\end{cases}}$} & \textbf{2}\tabularnewline
\hline 
Maple default & \textsl{$\frac{w}{\sqrt{w^{2}}}-\left(-1\right)^{\frac{1}{2}\frac{\mathrm{argument}\left(w^{2}\right)-\mathrm{2\, argument}\left(z\right)}{\pi}}$} & \textbf{2}\tabularnewline
\hline 
Maxima default & \textsl{$\frac{w}{\sqrt{w^{2}}}-\left(-1\right)^{\frac{\mathrm{atan2}\left(\sin\left(2\,\mathrm{carg}\left(w\right)\right),\cos\left(2\,\mathrm{carg}\left(w\right)\right)\right)-2\,\mathrm{carg}\left(w\right)}{2\pi}}$} & \textbf{2}\tabularnewline
\hline
\end{tabular}
\par\end{centering}

\end{table}

\par\end{center}


\begin{thebibliography}{10}
\bibitem{Brown} Brown, W.S: On computing with factored rational expressions.
\textsl{Proceedings of EUROSAM '74, ACM SIGSAM Bulletin} 8 (3), pp.
26-34, 1974.

\bibitem{Carette}Carette, J., A canonical form for piecewise defined
functions, \textsl{Proceedings of ISSAC 2007}, pp. 77-84.

\bibitem{CorlessAndJeffreyProvisos}Corless, R.M., Jeffrey, D.J.,
Well \ldots{} It isn't quite that simple. \textsl{ACM SIGSAM Bulletin}
26 (3), pp. 2-6, 1992.

\bibitem{CorlessAndJeffrey}Corless, R.M. and Jeffrey, D.J., Editor's
corner: The unwinding number, \textsl{ACM Communications in Computer
Algebra} 30 (2), pp. 28-35, 1996.

\bibitem{Jeffrey}Jeffrey, D.J., Branching out with inverse functions,
2009,\\
 \url{http://www.activemath.org/workshops/MathUI/09/proc/}

\bibitem{MosesSimplification} Moses, J: Algebraic simplification,
a guide for the perplexed. \textsl{Proceedings of the second ACM symposium
on symbolic and algebraic manipulation}, pp. 282-304, 1971

\bibitem{PappAndVizvari}Papp, D. and Vizvari, B: Effective solution
of linear Diophantine equation systems with an application to chemistry,
\textsl{Journal of Mathematical Chemistry} 39 (1), pp. 15-31, 2006.

\bibitem{RichAndJeffrey}Rich, A.D. and Jeffrey, D.J., Function evaluation
on branch cuts, \textsl{Communications in Computer Algebra} 30 (2),
pp. 25-27, 1996.

\bibitem{StoutemyerUsefulNumbers}Stoutemyer, D.R., Useful computations
need useful numbers, \textsl{ACM Communications in Computer Algebra
}41 (3), pp. 75-99, 2007.

\bibitem{Stoutemyer10Commandments}Stoutemyer, D.R., Ten commandments
for good default expression simplification, \textsl{Journal of Symbolic
Computation}, 46 (7), pp. 859-887, 2011.

\end{thebibliography}
\end{document}